\documentclass[11pt]{article}
\usepackage{amsfonts}
\usepackage{amsthm}
\usepackage{amssymb}
\usepackage[dvips]{graphicx,color}
\usepackage{graphicx}
\usepackage[margin=1.3in]{geometry}
\usepackage{xcolor}
\usepackage[utf8]{inputenc}
\usepackage{amsmath}
\usepackage{indentfirst}
\usepackage[dvips]{graphicx,color}
\usepackage{authblk}
\usepackage{enumitem}

\newcommand{\cambios}[1]{#1}

\def\be{\begin{eqnarray}}
\def\ee{\end{eqnarray}}
\def\beq{\begin{equation}}
\def\eeq{\end{equation}}

\def\({\left (}
\def\){\right )}

\graphicspath{{Figuras/}}
\newtheorem{theorem}{Theorem}[section]

\newtheorem*{note}{Note}
\newtheorem{proposition}[theorem]{Proposition}
\newtheorem{corollary}[theorem]{Corollary}
\newtheorem{definition}[theorem]{Definition}
\newtheorem{remark}[theorem]{Remark}

\title{A novel notion of null infinity for c-boundaries and generalized black holes}
\author[1]{I.P. Costa e Silva}
\author[2]{J.L. Flores}
\author[3]{J. Herrera}
\affil[1]{\small{{\textit{Department of Mathematics\\
Universidade Federal de Santa Catarina, 88.040-900 Florian\'{o}polis-SC, Brazil.}}}}
\affil[2]{\small{\textit{Departamento de \'Algebra, Geometr\'{i}a y Topolog\'{i}a\\ Facultad de Ciencias, Universidad de M\'alaga\\ Campus Teatinos, 29071 M\'alaga, Spain.}}}
\affil[3]{\small{\textit{Departamento de Matemáticas, Edificio Albert Einstein, Universidad de Córdoba\\ Campus de Rabanales, 14071 Córdoba, Spain}}}
\begin{document}

\maketitle

	\begin{abstract}
\noindent
  We give new definitions of null infinity and black hole in terms of causal boundaries, applicable to any strongly causal spacetime $(M,g)$. These are meant to extend the standard ones given in terms of conformal boundaries, and use the new definitions to prove a classic result in black hole theory for this more general context: if the null infinity is {\em regular} (i.e. well behaved in a suitable sense) and $(M,g)$ obeys the null convergence condition, then any closed trapped surface in $(M,g)$ has to be inside the black hole region. As an illustration of this general construction, we apply it to the class of {\em generalized plane waves}, where the conformal null infinity is not always well-defined. In particular, it is shown that (generalized) black hole regions do {\em not} exist in a large family of these spacetimes.

\end{abstract}

\section{Introduction}

Since its introduction by Geroch, Kronheimer and Penrose in 1972 \cite{GerochIdealPointsSpaceTime1972}, the {\em causal boundary} (or {\em $c$-boundary}) of a given strongly causal spacetime $(M,g)$ has asserted itself as a method for attaching ideal points to $M$ which, although somewhat more abstract than conformal boundaries, is (unlike the latter) deeply ingrained in the very causal structure of $(M,g)$. In particular, this means that the $c$-boundary is, in a strong sense, an inescapable feature of the (conformally invariant aspects of the) spacetime geometry.

By contrast, from a purely geometric perspective, the existence of a conformal boundary seems to impose an {\em ad hoc} restriction on $(M,g)$: it relies on the existence of a suitable open conformal embedding into a larger spacetime $(\tilde{M},\tilde{g})$, inducing a piecewise $C^1$ (often required to be $C^{\infty}$) boundary $\partial M$ with fairly restrictive properties first introduced by Roger Penrose in 1963 as a geometric description of the asymptotic properties of the gravitational field \cite{PenroseAsymptoticStructure1963,PenroseConformalInfinity1964} (see also, e.g., \cite{HawkingLargeScaleStructure1975,WaldGeneralRelativity1984,Frauendiener2004} and their references for general accounts and further technical details and \cite{ChruscielConformalboundaryextensions2010} for a analysis of the uniqueness and some criteria for existence of conformal boundaries). It is unclear to us how generic are such properties, and indeed the very existence, of conformal boundaries.

To be sure, it is well-known that among solutions of the Einstein field equations with various signs of the cosmological constant, the seminal work of Friedrich \cite{Friedrich1,Friedrich2,Friedrich3} shows that the existence of a conformal {\em future null infinity} ${\cal J}^+$ in the sense of Penrose is a {\em stable} feature under (non-linear) perturbations for a large class of initial data for the Einstein field equations, either vacuum or suitable matter content \cite{Friedrich2,ChruschielDelayExistenceAsymptoticallySimple2002}. However, even in simple but key examples such as plane waves, which also have great physical interest, the conformal boundary may either not exist altogether, or else have counterintuitive properties \cite{Florescausalboundarywavetype2008}. Although such solutions are indeed fairly special and probably too idealized from a physical standpoint, to the best of our knowledge the absence of a sensible conformal boundary they feature might well turn out to be a stable feature by itself (see, for example, \cite{Friedrich3} and references therein for an extensive discussion of this difficult issue).

Be as it may, it cannot be gainsaid that conformal boundaries are extremely convenient in physical applications. Whenever available, they present an exceedingly elegant and well-motivated means of addressing a number of key physical notions. One of the most beautiful applications of conformal boundaries in physics is in the classical theory of black holes. These are often informally described as spacetime regions from which causal communication with ``external observers'' is barred. Now, since the precise mathematical surrogate of ``causal communication'' is fairly clear in Lorentzian geometry, any purported rigorous definition of black holes is tantamount to giving a mathematically accurate and physically useful model of the term ``external observer''.

In standard treatments, these ``observers'' have been equated with the points on the future null infinity ${\cal J}^+$ in the sense of Penrose indicated above. (Indeed, the future part of conformal boundary is often taken to consist only of such points.) However, as also pointed out before, the properties of the conformal null infinity which are useful in mathematical Relativity can be rather specialized, abstracted as they are from very specific spacetimes which are solutions of the Einstein field equation.

Therefore, it can be of great interest to take advantage of the universality and geometric naturalness of the causal boundary construction, while making it as amenable as possible to concrete applications. The original causal boundary construction \cite{GerochIdealPointsSpaceTime1972} presented a number of undesirable properties even in the simplest cases. Fortunately, the main glitches in the original proposal have been addressed and solved by a suitable redefinition of the $c$-boundary, initially by Szabados \cite{SzabadosCausalboundarystrongly1988,SzabadosCausalboundarystrongly1989} and Marolf and Ross \cite{Marolfnewrecipecausal2003}, and further streamlined in a quite satisfactory way by Flores in \cite{Florescausalboundaryspacetimes2007}, and by Flores, Herrera and S\'{a}nchez in \cite{Floresfinaldefinitioncausal2011} (see also the references of the latter paper and in \cite{Garcia-ParradoCausalstructurescausal2005} for a detailed discussion of previous/alternative proposals by other authors).

Our purpose in this paper is to show how one can use the $c$-boundary construction to give a completely general definition of future null infinity ${\cal J}^+$ which is purely geometric, and valid for any strongly causal spacetime $(M^{n+1},g)$ with $n \geq 2$. Since the definition we give relies only on the $c$-boundary, one does not need to assume any field equations, specific dimensions or special asymptotic properties of $(M,g)$. 
Still, our approach here strongly suggests that, at least in many concrete physical situations in which the conformal null infinity plays a role, the latter will coincide with the $c$-boundary version. If so, we do not loose any of the good properties of the conformal boundaries (see \cite{CostaFloresHerrera2018-2}).

We also use our notion of null infinity to define the concomitant notion of a black hole in this general context. It is a matter of no little surprise that some non-trivial properties of black holes can indeed be extended to this general setting. A simple but important example we discuss here is a classic result in black hole theory (see, e.g., Prop. 12.2.2 of \cite{WaldGeneralRelativity1984}: under physically reasonable geometric conditions on the spacetime, any closed trapped surface remains inside the black hole region thereof. Other minor structural features characterizing black holes are also discussed here and shown to apply even in this very general situation.

It should be emphasized that these extended notions are not meant to provide {\em substitutes} for the standard, conformal notions which are so useful in physical applications, but rather {\em complements} to them. Finer aspects of physical interest, such as the detailed conditions characterizing (null) asymptotic flatness, or key results like the so-called ``peeling property'' of the Weyl tensor in gravitational wave theory \cite{SachsPeeling1961} are best described in terms of conformal boundaries, precisely because the definition of the latter structure has been fine-tuned for a long time to achieve this. But an array of extra assumptions have to be added, albeit physically tested and well-motivated. Indeed, the merest glance at the definition of asymptotic flatness in Wald's book \cite{WaldGeneralRelativity1984}, for example, reveals the number of assumptions one has to make for the notion of infinity to work at its fullest power.

The point here is that even in much less stringent situations we may still hope for usefully doing as much with less. The results about black holes proved here offer a case in point, but we also provide examples (in section \ref{ppwaves}) which are of considerable physical interest (indeed, a class of the so-called {\em generalized plane waves}) but where conformal notions are out of the question, since a conformal boundary (with good properties) simply may not exist!

The generality we envisage here, however, comes at a price. In order to obtain good results from our notions of null infinity and black holes, a number of {\em seemingly} rather technical assumptions and definitions are needed along the way. This is because a number of annoying counterexamples may be found for most of the otherwise intuitive statements which hold in the conformal case. We have endeavoured to present these extra hypotheses clearly, and while it may not at all be obvious to the reader, the specific assumptions we adopt are actually very mild and hold for all but some rather artificial examples. However, including a detailed discussion here would enormously complicate the paper, so we analyze them in depth elsewhere \cite{CostaFloresHerrera2018-2}. In particular it will be shown on that juncture that most assumptions in this paper are automatically satisfied for a large class of spacetimes where an interesting conformal boundary exists.

The rest of this paper is organized as follows.

In section \ref{prelim}, we introduce the notation and briefly review some of the basic definitions and results about the $c$-boundary of a strongly causal spacetime which will be used in the later parts of the paper.



 In section \ref{conf}, we briefly review the results in \cite{Floresfinaldefinitioncausal2011} about the relation between the c-boundaries and the conformal ones (when the latter are suitably defined). 


 In section \ref{bh}, a definition is given of null infinity for $c$-boundaries of strongly causal spacetimes. This definition is directly inspired by the classical notion of null infinity given in terms of the conformal boundary, plus the needed extra technicalities. We also discuss therein some of its general properties.

This extended notion of null infinity is then used in section \ref{bh2} to introduce a general notion of black hole, and we explore some of its more elementary properties. It is here that, under very mild assumptions 
(taken for granted in the usual, conformal setting \cite{CostaFloresHerrera2018-2}) which {\em do not} include any notion of asymptotic flatness, we are able to prove that any closed trapped surface in $(M,g)$ has to be inside the black hole region (Corollary \ref{trappedcor}), which as mentioned before is a key result in (standard) black hole theory.

In section \ref{ppwaves} the previous results are applied to give a description of null infinity for a class of generalized plane waves, which may {\em not} admit conformal boundaries, and it is shown that these spacetimes do not admit black holes (see \cite{Florescausalboundarywavetype2008,HubenyRagamaniNoHorizons,SenovillaNoBHinPPWaves2003} for related results).In Ref. \cite{SenovillaNoBHinPPWaves2003}, especially, the author shows that under certain mild additional geometric restrictions on generalized wave-type spacetimes, closed trapped surfaces do not exist therein\footnote{We thank Jos\'{e} M. M. Senovilla for bringing this reference to our attention.}. That (purely quasi-local) result is interestingly complementary to our own results here: in the context of Corollary \ref{trappedcor} we too may conclude, from a different perspective, that insofar as any trapped surfaces would be inside a black hole region, they do not exist in the wave-type spacetimes considered here.

Finally, in an appendix we discuss an example which shows that the clauses in Definitions \ref{scri} and \ref{ample} are not only natural when interpreted from the classical conformal approach, but cannot be discarded in favor of less technical-looking ones.



\section{Preliminaries: the c-boundary construction}\label{prelim}

Throughout this paper, $(M,g)$ will denote a {\it spacetime}, i.e. a connected, time-oriented $C^{\infty}$ Lorentzian manifold $(M,g)$. We shall assume that the reader is familiar with standard facts in Lorentzian geometry and causal theory as given in the basic references \cite{BeemGlobalLorentzianGeometry1996,HawkingLargeScaleStructure1975,ONeillSemiRiemannianGeometryApplications1983}.

The causal completion of $(M,g)$ is by now a standard construction put forth for the first time in a seminal paper by Geroch, Kronheimer and Penrose \cite{GerochIdealPointsSpaceTime1972}. The underlying idea is simple enough: to add \textit{ideal points} to the original spacetime, comprising the so called \textit{causal completion} (or \textit{$c$-completion} for short), in such a way that any inextendible timelike curve in $(M,g)$ has endpoints. The set of these ideal endpoints form the associated {\em causal boundary}, or {\em $c$-boundary} for short. This construction is in addition conformally invariant in the sense that two spacetimes in the same conformal class have identical $c$-boundaries. Finally, it is applicable to any strongly causal spacetime, independently of their asymptotic properties.

For the convenience of the reader and to set up the notation we shall use throughout, we briefly review the c-boundary construction here, referring to \cite{GerochIdealPointsSpaceTime1972} and \cite{Floresfinaldefinitioncausal2011} for further details and proofs. Let us start with some basic
notions.

A non-empty subset $P\subset M$ is called a {\em past set} if it
coincides with its chronological past, i.e. $P=I^{-}(P)$. (In particular, every past set is open.) A past set that cannot be written as the
union of two proper subsets, both of which are also past sets, is
said to be an {\em indecomposable past} set (IP). It can be shown that an IP either coincides with the past of some point of the spacetime, i.e., $P=I^{-}(p)$ for $p\in M$, or else $P=I^{-}(\gamma)$ for some inextendible future-directed
timelike curve $\gamma$. In the former case, $P$ is said to be a {\em proper indecomposable past
set} (PIP), and in the latter case $P$ is said to be a {\em
terminal indecomposable past set} (TIP). These two classes of IPs are disjoint.

Another useful technical definition is the following. The {\em common past} of a given set $S\subset
M$ is defined by \[\downarrow S:=I^{-}(\{p\in M:\;\; p\ll
q\;\;\forall q\in S\}).\]
The corresponding
definitions for {\em future sets}, IFs, TIFs,
PIFs, {\em common future}, etc., are obtained just by interchanging the roles of past and
future, and will always be understood.

The set of all IPs constitutes the so-called {\it future $c$-completion} of $(M,g)$, denoted by $\hat{M}$. If $(M,g)$ is strongly causal, then $M$ can naturally be viewed as a subset of $\hat{M}$ by identifying every point $p\in M$ with its respective PIP, namely $I^-(p)$. Indeed, it is well-known \cite{BeemGlobalLorentzianGeometry1996} that a strongly causal $(M,g)$ is {\em distinguishing}, i.e.,
\[
\forall p,q \in M, \, I^{\pm}(p)=I^{\pm}(q) \Longrightarrow p=q,
\]
so that the inclusion $i: p\in M \hookrightarrow I^-(p) \in \hat{M}$ is indeed one-to-one.

{\em Throughout this paper, unless otherwise explicitly stated, we shall assume that $(M,g)$ is strongly causal.}

The {\it future $c$-boundary} $\hat{\partial} M$ of $(M,g)$ is defined as the set of all its TIPs. Therefore, upon identifying $M$ with its image in $\hat{M}$ by the natural inclusion as outlined above,
\[
\hat{\partial} M \equiv \hat{M} \setminus M.
\]
The definitions of \textit{past $c$-completion} $\check{M}$ and {\it past $c$-boundary} $\check{\partial}M$ of $(M,g)$ are readily defined in a time-dual fashion using IFs.

Now, one would expect that the {\em (total) $c$-completion} of $(M,g)$ could be obtained by ``joining together'' in some suitable sense the past and future $c$-completions. However, naive attempts to do so meet with some surprisingly thorny technical issues (again, see \cite{Floresfinaldefinitioncausal2011} and references therein for a discussion). To circumvent these issues, certain clever manipulations are required which were first carried out by Szabados \cite{SzabadosCausalboundarystrongly1988, SzabadosCausalboundarystrongly1989}, with important improvements by Marolf and Ross \cite{Marolfnewrecipecausal2003}.

First, we introduce the so-called {\em Szabados relation} (or \textit{$S$-relation}) between IPs and IFs: an IP $P$ and an IF $F$ are {\em S-related}, $P\sim_{S}F$, if $P$ is a maximal IP inside
$\downarrow F$ and $F$ is a maximal IF inside $\uparrow P$. In particular, for any $p \in M$, it can be shown that $I^-(p) \sim_{S} I^+(p)$. Then, we have the following

\begin{definition}\label{d1} The {\em (total) c-completion} $\overline{M}$ is
formed by all
the pairs $(P,F)$ formed by $P\in \hat{M}\cup\{\emptyset\}$ and
$F\in \check{M}\cup\{\emptyset\}$ such that either
\begin{itemize}
\item[i)] both $P$ and $F$ are non-empty and $P\sim_{S}F$; or
\item[ii)] $P=\emptyset$, $F \neq \emptyset$ and there is no $P'\neq \emptyset$ such that $P'\sim_{S}F$; or
\item[iii)] $F=\emptyset$, $P \neq \emptyset$ and there is no $F'\neq\emptyset$ such that $P\sim_{S}F'$.
\end{itemize}
The original manifold $M$ is then identified with the set $\{(I^{-}(p),I^{+}(p)): p\in M\}$, and the {\em c-boundary} is defined as $\partial M\equiv\overline{M}\setminus M$.
\end{definition}
\begin{remark}\label{r0} \emph{We will systematically use the following fact, which is easy to check: given any IP $P \neq \emptyset$ (resp. any IF $F \neq \emptyset$), there always exists $F \in \check{M}\cup\{\emptyset\}$ (resp. $P \in \hat{M}\cup\{\emptyset\}$) such that $(P,F) \in \overline{M}$.
}
\end{remark}

Having defined the set structure of the $c$-completion, the next step is to extend the chronological relation in $(M,g)$ to the $c$-completion as follows:

\begin{equation}
(P,F)\ll
(P',F')\;\;\iff\;\; F\cap P'\neq\emptyset. \label{eq:7}
\end{equation}

Moreover, two pairs $(P,F)$,
$(P',F')$ are {\em causally related}, denoted by $(P,F)\leq (P',F')$, if
$F'\subset F$ and $P\subset P'$. Finally, these pairs said to be {\em
horismotically related} if they are causally, but not
chronologically related\footnote{\label{fn:1}These notions for causal and horismotical relations are not totally satifactory in general (see for instance \cite[Section 3.2]{Marolfnewrecipecausal2003} for more discussion on this issue). However, under the additional hypotheses we will assume later in this paper, they can be adopted without problems.}.

In concrete applications, it is also important to introduce a suitable topology on the $c$-completion $\overline{M}$.
The original topology considered in \cite{GerochIdealPointsSpaceTime1972}, although Hausdorff, was plagued by a number of technical problems and failed to yield sensible results even in simple cases. (An extensive discussion of these issues with further pertinent references on alternative proposed topologies on the $c$-completion can be found in Ref. \cite{Floresfinaldefinitioncausal2011}.) Our choice here, for reasons which are exhaustively discussed in \cite{Floresfinaldefinitioncausal2011}, is the so-called {\em chronological topology} on the $c$-completion.

The chronological topology is more conveniently defined by means of a so-called {\em limit operator} (see for instance \cite{FloresHausdorffseparabilityboundaries2016}). Let us briefly recall this notion here. Let $X$ be any set. Denote by ${\cal S}_X$ the set of all (infinite) sequences in $X$ (including their subsequences), and by $\mathbb{P}(X)$ the power set of all subsets of $X$. A {\em limit operator} on $X$ is any mapping $L:{\cal S}_X \rightarrow \mathbb{P}(X)$ such that if $\sigma \in {\cal S}_X$ is a sequence in $X$ and $\sigma '$ is a subsequence of $\sigma$, then $L(\sigma) \subset L(\sigma ')$. If $L$ is one such limit operator, the associated {\em derived topology} $\tau _{L}$ is defined via its closed sets: by definition, a subset $C \subset X$ is closed in $\tau _L$ if and only if $L(\sigma) \subset C$ for any sequence $\sigma$ of elements of $C$. The following facts about the topological space $(X,\tau _L)$ thus defined are readily verified.
\begin{itemize}
\item[LO1)] If $x \in L(\sigma)$ for a sequence $\sigma \in {\cal S}_X$, then $\sigma$ converges to $x$ with respect to $\tau_L$.
\item[LO2)] The topological space $(X,\tau _L)$ is {\em sequential}, i.e., a set $C \subset X$ is closed if and only if it is sequentially closed therein.
\end{itemize}
A limit operator $L$ is said of {\em first order} if the converse of (LO1) also holds, i.e. if the following equivalence is satisfied:
  \begin{equation}
x\in L(\sigma)\iff\sigma \hbox{ converges to $x$ with respect to $\tau_{L}$}.\label{eq:8}
\end{equation}

Recall also that with any sequence $\{A_n\}_{n \in \mathbb{N}}$ of subsets of $X$ we can associate the {\em Hausdorff inferior and superior limits} of sets as
\begin{eqnarray}
\mathrm{LI}(A_{n})\equiv
\liminf(A_{n})&:=&\cup_{n=1}^{\infty}\cap_{k=n}^{\infty}A_{k} \\
\mathrm{LS}(A_{n})\equiv
\limsup(A_{n})&:=&\cap_{n=1}^{\infty}\cup_{k=n}^{\infty}A_{k}.
\end{eqnarray}
Clearly one always has $\mathrm{LI}(A_{n}) \subset \mathrm{LS}(A_{n})$. Moreover, if $\{A_m\}$ is any subsequence,
\[
\mathrm{LI}(A_{n}) \subset \mathrm{LI}(A_{m}) \subset \mathrm{LS}(A_{m}) \subset \mathrm{LS}(A_{n}),
\]
Simple examples show that these inclusions are usually strict.

Next, we define the {\em future chronological limit operator} $\hat{L}$ on $\hat{M}$ as follows. Given a sequence $\sigma=\{P_{n}\}_{n}\subset \hat{M}$ of IPs and $P\in \hat{M}$, we set

\begin{equation}
\label{eq:3}
P\in \hat{L}(\sigma)\iff \left\{\begin{array}{l}
                                 P\subset \mathrm{LI}(\sigma)\\
                                 P \hbox{ is a maximal IP in }\mathrm{LS}(\sigma).
\end{array}\right.
\end{equation}
Again, by simply interchanging past and future sets we may analogously define the {\it past chronological limit operator} $\check{L}$ on $\check{M}$. Then, the {\em future (resp. past) chronological topology on $\hat{M}$ (resp. $\check{M}$)} is the derived topology associated to the limit operator $\hat{L}$ (resp. $\check{L}$).

\smallskip

We are now ready to define the chronological topology on the full $c$-boundary. To this end, define a limit operator $L$ on $\overline{M}$ as follows: given a sequence
$\sigma=\{(P_{n},F_{n})\}\subset\overline{M}$, put

\begin{equation}
\label{eq:4}
(P,F)\in L(\sigma)\iff \left\{
  \begin{array}{l}
    P\in \hat{L}(P_{n}) \hbox{ if $P\neq \emptyset$}\\
F\in\check{L}(F_n) \hbox{ if $F\neq \emptyset$}.
  \end{array}
\right.
\end{equation}
By definition, the {\em chronological topology on $\overline{M}$} is the derived topology $\tau_L$ associated to the limit operator $L$ defined in (\ref{eq:4}).

\begin{remark}\emph{
    In general, the limit operator defined in (\ref{eq:4}) is not of first order (see, for instance, \cite[Figure 7]{FloresHausdorffseparabilityboundaries2016}). However, this property can be proven to hold under some mild and general hypotheses \cite{FloresGromovCauchycausal2013, AkeSpacetimecoveringscasual2017}, valid in most cases of physical interest. Accordingly, throughout this paper we will implicitly assume that this property holds, and so, the equivalence (\ref{eq:8}) will be systematically used.
}
\end{remark}

The following result, whose proof is given in \cite{Floresfinaldefinitioncausal2011}, summarizes the key properties of the chronological topology.

\begin{theorem}
\label{thm:mainc-completion}
Let $(M,g)$ be a strongly causal spacetime and consider its associated $c$-completion $\overline{M}$ endowed with the chronological relations and chronological topology defined in (\ref{eq:7}) and (\ref{eq:4}), respectively. Then, the following statements hold.

\begin{enumerate}[label=(\roman*)]
\item \label{item:mismatopologia}The inclusion $M\hookrightarrow \overline{M}$ is continuous, with an open dense image. In particular, the topology induced on $M$ by the chronological topology on $\overline{M}$ coincides with the original, manifold topology.
\item The chronological topology is second-countable and $T_1$-separable, but not necessarily Hausdorff.
\item \label{item:chains} Let $\{x_n\}\subset M$ be a {\em future (resp. past) chain}, i.e., a sequence satisfying that $x_n\ll x_{n+1}$ (resp. $x_{n+1}\ll x_{n}$) for all $n$. Then,
\[
\begin{array}{c}
L(\left\{ x_{n} \right\})=\left\{ (P,F)\in \overline{M}: P=I^-(\left\{ x_n \right\}) \right\}
\\
\hbox{(resp. $L(\left\{ x_{n} \right\})=\left\{ (P,F)\in \overline{M}: F=I^+(\left\{ x_n \right\}) \right\})$.}
\end{array}
\]

\item \label{item:c-completioncompleta} The c-completion is {\em complete} in the following sense: given any (future or past) chain $\{x_n\}\subset M$, necessarily $L(\left\{ x_{n} \right\}) \neq \emptyset$, i.e. any (future or past) chain converges in $\overline{M}$ (cf. (\ref{eq:8})).

\item The sets $I^{\pm}((P,F),\overline{M})$ are open for all $(P,F)\in \overline{M}$.
\end{enumerate}

 \end{theorem}

 Let $\gamma:[a,b) \rightarrow M$ be an arbitrary future(resp. past)-directed causal curve. We want to extend the usual notion of {\em future (resp. past) endpoint} of $\gamma$ to points on $\overline{M}$. Concretely, we say that a pair $(P,F) \in \overline{M}$ is a {\em future (resp. past) endpoint} of $\gamma$ if for any increasing sequence $\left\{ t_{n} \right\} \subset [a,b)$ with $t_{n}\nearrow b$, we have $(P,F) \in L(\left\{ \gamma (t_{n}) \right\})$.

 The following immediate consequence of Theorem \ref{thm:mainc-completion} describes when a point $(P,F)\in\overline{M}$ is an endpoint of a causal curve, and will be very important in later sections. We shall often deal only with future endpoints, since the corresponding statements for past endpoints are easy to obtain from time-duality and will always be understood.

\begin{corollary}
\label{cor:endpoints}
Suppose $(M,g)$ is strongly causal spacetime with $c$-completion $\overline{M}$ as above. Let $\gamma:[a,b) \rightarrow M$ be a future-directed causal curve in $M$ and $(P,F) \in \overline{M}$. Then, the following statements hold.
\begin{enumerate}[label=(\roman*)]
\item \label{item:endpoints1}If $P=I^-(p)$ and $F=I^+(p)$ for some $p \in M$, i.e., if $(P,F)$ is a point of $M$ via its natural inclusion, then $(P,F)$ is a future endpoint of $\gamma$ (in the extended sense) if, and only if, for any increasing sequence $\left\{ t_{n} \right\} \subset [a,b)$ with $t_{n}\nearrow b$, we have $\gamma (t_{n}) \rightarrow p$ in $M$. (In order words, $(P,F)$ is a future endpoint of $\gamma$ in the extended sense iff $p$ is an endpoint in the ordinary, spacetime sense.)
\item \label{item:endpoints2} If $\gamma$ is a future (resp. past) directed {\em timelike} curve, then
\begin{equation}
\label{keyeq1}
(P,F)\in\overline{M} \mbox{ is a future endpoint of $\gamma$} \Longleftrightarrow P=I^-(\gamma).
\end{equation}
Moreover, when $\gamma$ is timelike, then it always has some endpoint on $\overline{M}$.
\item \label{item:endpoint3}If the future-directed causal curve $\gamma$ has $(P,F)$ as future endpoint with $P\neq \emptyset$, then $P=I^{-}(\gamma)$.
\end{enumerate}
\end{corollary}

\begin{proof}
$(i)$ By definition, $(P,F)$ is a future endpoint of $\gamma$ (in the extended sense) if, and only if, for any increasing sequence $\left\{ t_{n} \right\} \subset [a,b)$ with $t_{n}\nearrow b$,
\[
(P,F) \in L(\left\{ \gamma (t_{n}) \right\}) \Longleftrightarrow (I^-(\gamma (t_{n})),I^+(\gamma (t_{n}))) \rightarrow (P,F) \mbox{ in $\overline{M}$}\Longleftrightarrow \gamma (t_{n}) \rightarrow p \mbox{ in $M$},
\]
where we have used the definition of endpoint in the extended sense for the first equivalence, the fact that the limit operator is first-order (cf. (\ref{eq:8})) to obtain the second equivalence, and the fact that the induced topology on $M$ is the manifold topology according to item \ref{item:mismatopologia} of Thm. \ref{thm:mainc-completion} to get the third equivalence.
\\
$(ii)$ If $\gamma$ is timelike, then for any increasing sequence $\left\{ t_{n} \right\} \subset [a,b)$ with $t_{n}\nearrow b$ the corresponding sequence $\left\{ \gamma (t_{n})\right\}$ is a future chain in $M$, (\ref{keyeq1}) follows immediately from items \ref{item:chains} and \ref{item:c-completioncompleta} of Thm. \ref{thm:mainc-completion}. \\
$(iii)$ Suppose $(P,F)$ is a future endpoint of the causal curve $\gamma$ with $P\neq \emptyset$, and take again any increasing sequence $\left\{ t_{n} \right\} \subset [a,b)$ with $t_{n}\nearrow b$, so that $(P,F) \in L(\left\{ \gamma (t_{n}) \right\})$. According to the definition of the limit operator $L$, Eq. (\ref{eq:4}), we have that
\[
P \in \hat{L}(\left\{ I^-(\gamma (t_{n}) \right\}),
\]
which in turn means, according to Eq. (\ref{eq:3}), that
\[
P\subset \mathrm{LI}(\left\{ I^-(\gamma (t_{n}) \right\}) \subset I^-(\gamma).
\]
Now, clearly $I^-(\gamma)\subset \mathrm{LS}(\left\{ I^-(\gamma (t_{n}) \right\})$, and the maximality of $P$ as an IP
therein, as required by (\ref{eq:3}), now yields that $P=I^-(\gamma)$, as desired.
\end{proof}

\begin{remark}\label{r}{\em The situation when $\gamma$ is causal but not timelike is more involved. The converse of item \ref{item:endpoint3} is no longer true (see \cite[Section 3.5]{Floresfinaldefinitioncausal2011}). Indeed, inextendible causal curves may not have endpoints in the c-completion. It is therefore natural to consider the following definition (compare with  \cite[Definition 3.33]{Floresfinaldefinitioncausal2011}).}
\end{remark}
\begin{definition}\label{def:properly-causal}
  A $c$-completion $\overline{M}$ is said to be {\em properly causal} if any future or past-inextendible causal curve in $M$ has an endpoint on $\overline{M}$. \end{definition}
\noindent Some conditions ensuring the properly causal property were summarized in \cite[Theorem 3.35]{Floresfinaldefinitioncausal2011}. For instance, a spacetime is properly causal if it is globally hyperbolic. In order to obtain more general results, we will take a special look at \cite[Theorem 3.35 (iii)]{Floresfinaldefinitioncausal2011}. In fact, that theorem motivates the following definition.

\begin{definition} A future-inextendible (resp. past-inextendible) causal curve $\gamma:[a,b)\rightarrow M$ is {\em future-regular} (resp. {\em past-regular}) if
 \[
 \uparrow \gamma = \uparrow I^{-}(\gamma)\quad
    \left({\rm resp.} \downarrow \gamma = \downarrow I^{+}(\gamma) \right).
\]
 \end{definition}
\begin{remark}\label{r2}{\em
  Clearly, the inclusion $\uparrow \gamma \subset \uparrow I^{-}(\gamma)$ holds for any future-inextendible causal curve $\gamma$. Moreover, any future-inextendible timelike curve is future-regular (with time-dual analogous statements). In a globally hyperbolic spacetime the common future $\uparrow I^-(\gamma)$ of any future-directed causal curve $\gamma$ is empty. In particular, any future-inextendible causal curve is future-regular. Analogously, past-inextendible causal curves are always past-regular in a globally hyperbolic spacetime. On the other hand, there are very simple examples where future(past)-regularity fails. For instance, take the flat $2d$ Minkowski spacetime $(\mathbb{R}^2, -dt^2+dx^2)$ and delete the spacelike half-axis $t=0,x\geq 0$. Then, the null geodesic generator of the past of the origin $(0,0)$ on the side of positive $x$ is not future-regular.}
\end{remark}
\noindent The following result shows that only a few very specific causal curves can fail to be regular (see also \cite[Prop. 3.2]{Florescausalboundarywavetype2008}.

  \begin{proposition}
    Let $\gamma:[a,b)\rightarrow M$ be a future-inextendible (resp. past-inextendible) causal curve. If there is no $c\in [a,b)$ such that $\gamma|_{(c,b)}$ is a null geodesic ray, then $\gamma$ is future(resp. past)-regular.
  \end{proposition}
  \begin{proof}
     Let us focus on the future case (the past case is proved analogously). From the hypothesis, there exists a sequence $\left\{ t_{n} \right\}$ with $t_{n}\nearrow b$ such that $\gamma(t_{n})\ll \gamma(t_{n+1})$ for all $n$. Therefore, we can construct a timelike curve $\eta$ by concatenating timelike curves connecting consecutive points of the sequence. By construction, $\eta$ and $\gamma$ have the same past and the same common future, and the result follows from the fact that any timelike curve $\gamma$ is future-regular.
  \end{proof}

\noindent On the other hand, we shall often need to work with null geodesic rays; to make them treatable, we introduce the following definition:
  \begin{definition}
\label{def:sproperlycausal} \noindent A c-completion $\overline{M}$ is {\em strongly properly causal} if any future-inextendible (resp. past-inextendible) null geodesic ray is future(resp. past)-regular.
  \end{definition}
\noindent Clearly, strong proper causality implies proper causality (see \cite[Theorem 3.35 (iii)]{Floresfinaldefinitioncausal2011}). One might ask whether there is any natural causal condition on $(M,g)$ which implies strong proper causality. Note, however, that the flat $2d$ spacetime we construct in Remark \ref{r2} is {\em stably causal}, since the coordinate $t$ gives a time-function therein. Yet, it turns out that {\em causal continuity}, a causal condition immediately stronger than stable causality in the causal ladder (see Section 3.3 of Ch. 3 in \cite{BeemGlobalLorentzianGeometry1996}), holds in $M$:

\begin{proposition}\label{prop:causalcontinuity}
  The $c$-completion of any causally continuous spacetime $(M,g)$ is strongly properly causal.
\end{proposition}
\begin{proof}
Suppose, by way of contradiction, that $(M,g)$ is causality continuous but there exists a future-inextendible null geodesic ray $\gamma:[a,b) \rightarrow M$ which is not future-regular. Since we always have $\uparrow \gamma \subset \uparrow I^{-}(\gamma)$ (cf. Remark \ref{r2}), there exists $p \in \uparrow I^{-}(\gamma)\setminus \uparrow \gamma$. Let $q \in I^-(p)$ such that $I^{-}(\gamma) \subset I^-(q)$, and pick $q \ll q'\ll p$. Fix $t_0 \in [a,b)$. Since any sequence $\{x_n\}$ in $I^-(\gamma(t_0)) \subset I^-(q)$ such that $x_n \rightarrow \gamma(t_0)$ is in $I^-(q)$, we infer that $\gamma(t_0) \in \overline{I^-(q)}$. The causal continuity now implies (see Lemmas 3.42 and 3.43 of \cite{Minguzzicausalhierarchyspacetimes2008}) that $ q \in \overline{I^+(\gamma(t_0))}$, and hence $q' \in I^+(\gamma(t_0))$. Since $t_0$ is arbitrary, we conclude that $\gamma[a,b) \subset I^-(q')$. But then one concludes that $p \in \uparrow \gamma$, in contradiction.
\end{proof}

The importance of strong proper causality lies in that it provides the following characterization for the endpoints of any causal curve (which extends the equivalence (\ref{keyeq1}) for timelike curves):
\begin{proposition}\label{prop:strongfirst}
  Let $\overline{M}$  be the c-completion of a strongly causal spacetime $(M,g)$, and assume that $\overline{M}$ is strongly properly causal. A pair $(P,F)$ is endpoint of a future(resp. past)-directed, future(resp. past)-regular causal curve if, and only if, $P=I^{-}(\gamma)$ (resp. $F=I^{+}(\gamma)$).
\end{proposition}
\begin{proof}
Apply the same ideas as in the proofs of Theorem \ref{thm:mainc-completion} and \cite[Theor. 3.35 (iii)]{Floresfinaldefinitioncausal2011}.
\end{proof}


\section{Conformal extensions vs $c$-completion}
\label{conf}

The standard, mathematically precise definitions of {\em null infinity} and {\em black hole} are based on the notion of  conformal boundary of a spacetime as introduced by R. Penrose \cite{PenroseAsymptoticStructure1963,PenroseConformalInfinity1964}. We wish to generalize these notions to the context of $c$-boundaries, and accordingly we start by revisiting the relation between these two boundaries, which comprises some of the main results in \cite[Section 4]{Floresfinaldefinitioncausal2011}.

Henceforth, the $c$-completion $\overline{M}$ of $(M,g)$ will always be understood to be endowed both with the chronological relations and the chronological topology as defined in the previous section.


%

\begin{definition}\label{def:envelopment}
  Given a strongly causal spacetime $(M,g)$, a {\em conformal extension} of $(M,g)$ is an open embedding $i:(M,g) \hookrightarrow (\tilde{M},\tilde{g})$ of $(M,g)$ into some strongly causal spacetime $(\tilde{M},\tilde{g})$ preserving time-orientation, for which there exists a strictly positive function $\Omega\in C^{\infty}(M)$ satisfying
    \begin{equation}
\label{conformalfactor1}
i^{*}\tilde{g}=\Omega^2 g.
\end{equation}
$\Omega$ is called the {\em conformal factor} of the conformal extension $i$.

The {\em conformal completion} of $(M,g)$ with respect to the conformal extension $i$ is defined as the (topological) closure $\overline{M}_i:=\overline{i(M)}\subset \tilde{M}$, and the associated {\em conformal boundary} as the (topological) boundary $\partial_i M:=\overline{i(M)}\setminus i(M)$.

Finally, we denote by $\overline{M}_i^*$ (resp. $\partial_i^* M$) the set of all the {\em accessible points} of the conformal completion (resp. conformal boundary), that is, the set of points on $\overline{M}_i$ (resp. $\partial_i M$) which are endpoints\footnote{Here, and throughout this section, the term ``endpoint'' is taken in the usual spacetime sense.} of timelike curves contained in $M$.
\end{definition}

\begin{remark}\label{rmk1}
{\em A few comments about Definition \ref{def:envelopment} are in order.
\begin{enumerate}
\item Note that we {\em do not} require, in this definition, that a given conformal boundary should have any regularity; even if it is piecewise smooth, we do not demand that the conformal factor extends smoothly to the boundary. As it stands, the definition is too weak to be useful in most applications, and has to be supplemented according to specific needs. In particular, the definition we give here is much more general than the ones usually found in the literature, which require at least $C^1$ smoothness of the boundary, extendibility of the conformal factor to the boundary and a host of other properties (see, e.g., Ch. 11 of \cite{WaldGeneralRelativity1984}). We too shall add some extra assumptions in what follows, but will shall do so gradually as needed, and the added assumptions will still be fairly general and comprise most concrete strongly causal examples in the literature.
\item If $(M,g)$ is extendible as a strongly causal spacetime, i.e., it is {\em isometric} to a proper open subset of a larger strongly causal spacetime $(\tilde{M}, \tilde{g})$, then the latter gives a conformal extension of the former with conformal factor $\Omega \equiv 1$. To avoid this kind of triviality, one usually assumes, in concrete applications, that $(M,g)$ is in some sense ``maximal'' in the strongly causal class.
\item The ``larger'' spacetime used in a conformal extension may well be $(M,g)$ itself. As a standard example, consider the following. Let $M = \mathbb{R}^2$ with the flat metric
\[
g = -dudv,
\]
and the time orientation such that both $\partial_u$ and $\partial_v$ are future-directed null vector fields. This spacetime is geodesically complete, and hence inextendible. Define
\[
i: (u,v) \in M=\mathbb{R}^2 \mapsto (\arctan u,\arctan v) \in \mathbb{R}^2
\]
and consider the smooth function
\[
\label{conformalfactor2}
\Omega : (u,v) \in \mathbb{R}^2 \mapsto \cos u \cos v \in \mathbb{R}.
\]
Then it is easy to check that $(\mathbb{R}^2,g)$ is a (nontrivial) conformal extension of $(M,g)$ such that the image of $M$ by $i$ is the open square $Q:= (-\pi/2, \pi/2)^2$, with conformal factor $\Omega \circ i$. Note that $\Omega$ vanishes on the boundary of $Q$ in $\mathbb{R}^2$. Thus, $(M,g)$ is conformally extended via a mapping onto an open set of itself.
\item Let $i:(M,g) \hookrightarrow (\tilde{M},\tilde{g})$ be a conformal extension of $(M,g)$. Since $i$ is a diffeomorphism when viewed as a map onto its open image $i(M)\subset \tilde{M}$, consider its (smooth) inverse $i^{-1}: i(M) \rightarrow M$. Then $(i(M),(i ^{-1})^{\ast} g)$ is a spacetime on its own right, and moreover it is isometric to $(M,g)$ by construction. Hence, there is no loss of generality in regarding $M \subseteq \tilde{M}$ and $i$ to be the inclusion map. We shall often do so in what follows, and will then abuse of notation by referring to $(\tilde{M},\tilde{g})$ itself as a conformal extension and discard any reference to the map $i$. In this case, we write the condition (\ref{conformalfactor1}) as
\[
\label{conformalfactor3}
\tilde{g}|_{M} = \Omega^2 g.
\]

\item Even upon identifying isometric spacetimes as in the previous item, conformal extensions, and also the associated conformal boundaries, are clearly far from being unique. But exploiting its relationship with the $c$-boundary, as we shall see below, allows one to build its accessible part in an essentially unique fashion.
\end{enumerate}
}
\end{remark}

A given conformal completion will always be assumed to be endowed with the induced topology from $\tilde{M}$ and with a chronological (resp. causal) relation defined as follows: two points $p,q\in \overline{M}_i$ are chronologically (resp. causally) related, $p\ll_i q$ (resp. $p \leq _i q$) if there exists a smooth future-directed timelike (resp. causal) curve $\gamma:[a,b]\rightarrow \overline{M}_i$ from $p$ to $q$ with $\gamma|_{(a,b)}\subset M$ (see \cite[Section 4.1]{Floresfinaldefinitioncausal2011} for additional discussion on these choices).

\medskip

As one might expect, in general the $c$-completion differs from a given conformal one. In fact, despite its name, the conformal completion is {\em not} conformally invariant (see, for instance, \cite[Figure 10]{Floresfinaldefinitioncausal2011}). However, (and this is one of the main points of this section) under some additional conditions on the conformal extension, it is possible that both completions coincide. We shall presently discuss some of these conditions.

\begin{definition}\label{def:chronologically complete}
Let $(M,g)$ be a spacetime. A conformal extension $(\tilde{M},\tilde{g})$ is said to be {\em chronologically complete} if any timelike curve $\gamma:[a,b)\rightarrow M$ which is inextendible in $M$ has an endpoint $p$ on the associated conformal boundary.
\end{definition}
Clearly, this condition parallels the point \ref{item:c-completioncompleta} in Theorem  \ref{thm:mainc-completion}, and is meant to ensure that the conformal completion has ``enough points''. (Otherwise, by suitably deleting points in a conformal completion, we have a highly non-unique construction which will then become geometrically useless.)

Another important requirement is that, around the points $p$ of the (accessible) conformal boundary, the causal structure of $(\tilde{M},\tilde{g})$ is adapted to that of $(M,g)$. To formalize this idea, let us begin with the concept of {\em timelike deformable points} (see \cite[Definition 4.10]{Floresfinaldefinitioncausal2011})

\begin{definition}\label{def:deformable}
Consider a continuous curve $\gamma:[a,b]\rightarrow \overline{M}_i$ such that $\gamma|_{[a,b)}$ is a future-directed smooth timelike curve contained in $M$ and $\gamma(b)\in \partial_i M$. Then, $\gamma$ is {\em future deformably timelike} if there exists a neighborhood $U=\tilde{U}\cap \overline{M}_i$ of $\gamma(b)$ (where $\tilde{U}$ is an open set of $\tilde{M}$) such that $\gamma(a)\ll_i \omega$ for all $\omega\in U$. (The notion of {\em past deformably timelike} is analogous.)

  An accessible point $p\in \partial_i^* M$ is {\em timelike deformable} if all the TIPs and TIFs associated to $p$ (i.e., TIPs and TIFs defined by timelike curves on $M$ with endpoint $p$) are intersections with $M$ of the chronological pasts or futures in $(\tilde{M},\tilde{g})$ of timelike deformable curves in $\overline{M}_i$.
\end{definition}

The previous notion ensures only that the {\em chronological} relation of the extension is well-behaved with respect to that of $(M,g)$; the following definition extends such good behaviour to the {\em causal} relation.

\begin{definition}\label{def:transitive}
  An accessible point $p\in \partial_i^* M$ is {\em (locally) timelike transitive} if it admits a neighborhood $V=\overline{M}_i\cap \tilde{V}$ ($\tilde{V}$ open in $\tilde{M}$) such that for any $q,q'\in V$:
  \begin{itemize}
  \item $q\ll_i p \leq_i q' \Rightarrow q\ll_i q'$.

  \item $q\leq_i p \ll_i q'\Rightarrow q\ll_i q'$.
  \end{itemize}
\end{definition}
 If an accessible point $p\in \partial_i^* M$ satisfies both timelike deformability and timelike transitivity, then we say that $p$ is {\it regularly accessible}. If all the accessible boundary points are regularly accessible, we will just say that $\partial_i^* M$ itself is {\it regularly accessible}.

\smallskip

We are now ready to present the main result in \cite[Section 4]{Floresfinaldefinitioncausal2011}.

\begin{theorem}\cite[Theorem 4.16]{Floresfinaldefinitioncausal2011} \label{thm:causaltoconformal} Let $i:(M,g) \hookrightarrow (\tilde{M},\tilde{g})$ be a chronologically complete extension of $(M,g)$. If the accessible conformal boundary $\partial_i^* M$ is regularly accessible, then the accessible part $\overline{M}_i^*$ of the conformal completion and the $c$-completion $\overline{M}$ are equivalent in the following precise sense:
\begin{enumerate}[label=(\roman*)]
  \item \label{thmcausaltoconformal-defPsi} There exists a homeomorphism $\Psi:\overline{M}\rightarrow \overline{M}_i^*$, which maps boundary to boundary.
  \item \label{thmcausaltoconformal-chrniso} $\Psi$ is a chronological isomorphism (i.e., $\Psi$ and $\Psi^{-1}$  preserve the chronological relations).
  \end{enumerate}
\end{theorem}

\section{Null infinity on $c$-boundaries}\label{bh}

We wish to introduce a notion of null infinity on $c$-boundaries which remains valid even when a conformal boundary does not exist.
Again, we will start with some general considerations, which will motivate our main results and definitions, as well as set up the appropriate notation.

First, let us give a more precise description of what we shall mean by {\em null infinity} here in the case when conformal boundaries are present. We will try and capture only some minimal aspects of the original Penrose's definition \cite{PenroseAsymptoticStructure1963,PenroseConformalInfinity1964,Frauendiener2004}, without some of the more technical and specific assumptions needed in physical applications (such as those given, for example, in Refs. \cite{HawkingLargeScaleStructure1975,WaldGeneralRelativity1984}).

\begin{definition}
\label{conformalinfinity1}
Let $i: (M,g) \hookrightarrow (\tilde{M},\tilde{g})$ be a conformal extension of the spacetime $(M,g)$ with conformal boundary $\partial_i M$ and conformal factor $\Omega$. An accessible point $p \in \partial^{*}_i M$ is said to {\em be at infinity} if there exist an open set $\tilde{U} \ni p$ of $\tilde{M}$ and a $C^{\infty}$ extension\footnote{It is possible, and may be interesting for certain purposes, to consider a lower - say, $C^1$-differentiability class without deep changes in the main results. Nevertheless, since this would make the following discussion a little more cumbersome, we adopt for simplicity a stronger regularity assumption.} $\tilde{\Omega}$ of $\Omega$ to $M \cup \tilde{U}$ such that $\tilde{\Omega}(p) =0$. Such a point at infinity is said to be {\em regular} if in addition $d\tilde{\Omega}(p) \neq 0$. The set of all regular points at infinity is the {\em conformal null infinity} ({\em associated with the conformal extension} $(\tilde{M},\tilde{g})$) and will be denoted by ${\cal J}_c$. We will say that a point $p\in {\cal J}_c$ belongs to the {\em future} (resp. {\em past}) {\em null infinity}, denoted by ${\cal J}_c^+$, (resp. ${\cal J}_c^-$) if $p$ is future (resp. past) accessible (recall Defn. \ref{def:envelopment})
\end{definition}

\begin{remark}
\label{rmk3}
{\em The following points about Definition \ref{conformalinfinity1} are relevant.
\begin{itemize}
\item[1)] The notions of a point $p \in \partial^{*}_i M$ being at infinity and/or being regular there clearly do not depend on the particular extension $\tilde{\Omega}$ of the conformal factor.
\item[2)] The notion of a point being at infinity is {\em not} a conformally invariant one. For example, if $\Omega$ can be smoothly extended to a point $p \in \partial^{*}_i M$ but the extension does not vanish there, we may rescale $g$ by (the restriction of) a conformal factor $\omega^2 \in C^{\infty}(\tilde{M})$ which diverges at $p$ suitably slowly so that $\Omega/\omega$ is still $C^{\infty}$. The point $p$ then becomes a point at infinity with respect to $(\tilde{M}, \tilde{g})$ when the latter is viewed as a conformal extension of $(M,\omega^2g)$. Conversely, points which are at infinity for a certain spacetime may fail to be so for others in the same conformal class. Indeed, the familiar example of $2d$ hyperbolic space already illustrates these remarks: points at the boundary of the unit disc in $\mathbb{R}^2$ are at infinity if we view the Euclidean plane as a conformal extension of the Poincar\'{e} disc model of the hyperbolic plane, but not if viewed as a conformal (actually isometric) extension of the (restriction of the) Euclidean metric on the disc.
\item[3)] As an example of this definition, we may look back at the example in item 3. of Remark \ref{rmk1}. There, the set of points at infinity consists of the union of the two segments $u=\pi/2,\, v>-\pi/2$ and $v=\pi/2, \, u> -\pi/2$. The point $u=v= \pi/2$ is the only non-regular point. This point is not the future endpoint of any null geodesic in $(M,g)$, but all {\em regular} points at infinity are endpoints of null geodesic {\em rays}.
\end{itemize}
}
\end{remark}

\begin{theorem}
\label{conformalinfinity2}
Let $(\tilde{M},\tilde{g})$ be a conformal extension of the spacetime $(M,g)$ with conformal boundary $\partial_iM$ and conformal factor $\Omega$. Let $p \in \partial^{*}_iM$. Then the following statements hold.
\begin{itemize}
\item[i)] Suppose $p$ is at infinity, and let $\alpha:[0,A) \rightarrow M$ ($0<A \leq +\infty$) be a future-inextendible null geodesic in $(M,g)$ (analogously for past-directions). If $\alpha$ has a future endpoint at $p$ when viewed as a causal curve in $(\tilde{M}, \tilde{g})$, then $A=+\infty$ and $\alpha$ is future-complete. In other words, any null geodesic in $(M,g)$ with a future endpoint at infinity is future-complete.
\item[ii)] Conversely, suppose the conformal factor extends in a $C^{\infty}$ fashion to $M\cup \tilde{U}$ where $\tilde{U} \ni p$ is open in $\tilde{M}$. If $p$ is a future (resp. past) endpoint of a future-complete (resp. past-complete) null geodesic in $(M,g)$, then $p$ is at infinity.
\item[iii)] Suppose $(M,g)$ is inextendible (as a spacetime\footnote{Again, this result will still work {\em mutatis mutandis} if one only assumes that $(M,g)$ cannot be extended as a $C^1$ spacetime. On the other hand, there are trivial counterexamples to this result without the assumption of inextendibility, obtained by taking $(M,g)$ to be a suitable open subset of Minkowski with rough boundary.}). Then ${\cal J}_c$ is a smooth hypersurface in $(\tilde{M},\tilde{g})$.
\end{itemize}
\end{theorem}
\begin{proof} We shall present our arguments for future directions, since the past directed case is again analogous. \\
$(i)$ Using Definition \ref{conformalinfinity1} we can assume that the conformal factor extends in a $C^{\infty}$ fashion to $M\cup \tilde{U}$, where $\tilde{U} \ni p$ is open in $\tilde{M}$, and $\tilde{\Omega}(p)=0$. We shall, for notational simplicity, continue to refer to this extended conformal factor as $\Omega$. Denote by $\tilde{\nabla}$ the Levi-Civita connection on $(\tilde{M},\tilde{g})$. Using the fact that $\tilde{g}|_M = (\Omega|_M)^2 g$ and that $\alpha$ is a null geodesic in $(M,g)$, we get
\begin{equation}
\label{geodesic}
\tilde{\nabla}_{\dot{\alpha} }\dot{\alpha} = \nabla_{\dot{\alpha}}\dot{\alpha} + 2\frac{(\Omega \circ \alpha)'}{(\Omega \circ \alpha)} \dot{\alpha} \equiv f \cdot \dot{\alpha},
\end{equation}
where $f:[0,A) \rightarrow \mathbb{R}$ is the smooth function given by
\[
f := 2\frac{(\Omega \circ \alpha)'}{\Omega \circ \alpha}.
\]
Exercise 3.19 in p. 95 of \cite{ONeillSemiRiemannianGeometryApplications1983} now implies that $\alpha$ is a null pregeodesic in $(\tilde{M},\tilde{g})$. Let $h:[0,a) \rightarrow [0,A)$ ($0<a \leq +\infty$) be an increasing reparametrization of $\alpha$ such that $\hat{\alpha}:= \alpha \circ h$ is a geodesic in $(\tilde{M},\tilde{g})$. Since it has an endpoint in $p \in \partial^{*}_i M \subset \tilde{M}$, it is actually extendible as a null geodesic in $(\tilde{M},\tilde{g})$, and in particular $a <+\infty$. Denote by $\beta:[0,a] \rightarrow \tilde{M}$ the null $\tilde{g}$-geodesic segment extending $\hat{\alpha}$, with $\beta(a) \equiv p$. Since $\Omega(p) =0$, and $\Omega$ is $C^1$, by the mean value theorem there exists a number $k>0$ such that
\begin{equation}
\label{lipschitz}
(\Omega \circ \beta)(s) \leq k (a - s), \forall s \in [0,a].
\end{equation}
However, by exercise 3.19 of \cite{ONeillSemiRiemannianGeometryApplications1983} we must have
\[
h'' + (f \circ h) (h')^2 =0
\]
on $[0,a)$, which when we substitute the definition of $f$ gives, multiplying through $(\Omega \circ \alpha \circ h)^2$,
\begin{eqnarray}
0 &=& (\Omega \circ \alpha \circ h)^2 h'' + 2 (\Omega \circ \alpha \circ h)[(\Omega \circ \alpha)'\circ h)h ']h' \\ \nonumber
&=& (\Omega \circ \alpha \circ h)^2 h'' + 2 (\Omega \circ \alpha \circ h)(\Omega \circ \alpha \circ h)'h' \\ \nonumber
&=& (\Omega \circ \hat{\alpha})^2 h'' + 2(\Omega \circ \hat{\alpha})(\Omega \circ \hat{\alpha})'h' \\ \nonumber
&\equiv & [(\Omega \circ \hat{\alpha})^2 h']'.
\end{eqnarray}
Hence, for some constant $c \in \mathbb{R}$,
\begin{equation}
\label{relation}
(\Omega \circ \hat{\alpha})^2 h' = c.
\end{equation}
(We deduce that $c>0$, since $\Omega \circ \hat{\alpha} >0$ and $h$ is strictly increasing.) Therefore, for each $0\leq t <a$, we have
\begin{equation}
\label{integral}
h(t) = \int_0^t h'(s) ds = c \int_0^t \frac{ds}{[(\Omega \circ \hat{\alpha})(s)]^2} \geq \frac{c}{k^2} \int_0^t \frac{ds}{(a -s)^2},
\end{equation}
where we have used (\ref{lipschitz}) to get the last inequality. We thus conclude that for each $0\leq t <a$,
\[
A \geq h(t) \geq \frac{c}{k^2}\left( \frac{1}{a-t} - \frac{1}{a} \right).
\]
Since the right-hand side of this inequality diverges when $t \rightarrow a$, we conclude that $A \equiv +\infty$, proving $(i)$.

\smallskip

$(ii)$ Using the same notation as in $(i)$, assume, by way of contradiction, that $\alpha$ is future-complete in $(M,g)$, that is, that $A=+\infty$, but that $\Omega(p) \neq 0$. Then we may take $(\Omega \circ \hat{\alpha})$ to be bounded from below by some positive constant $B$ in Eq. (\ref{integral}), whence we conclude that
\[
h(t) \leq c/B, \, \forall t \in [0,a),
\]
and hence that $A \leq c/B$, a contradiction.

\smallskip

$(iii)$ Again, we may focus without loss of generality on the future part ${\cal J}^+_c$ of the conformal boundary. Let $p \in {\cal J}^+_c$. By definition, there exists an open neighborhood $\tilde{U}\ni p$ in $\tilde{M}$ to which $\Omega$ extends as a $C^{\infty}$ function (which we again denote by $\Omega$). Let $\tilde{\nabla}\Omega$ denote the (metrically) associated gradient vector field in $M\cup \tilde{U}$. Since $\tilde{\nabla}\Omega (p) \neq 0$, we can assume (shrinking it if necessary) that $\tilde{U}$ is a connected coordinate neighborhood contained in $I^+(M, \tilde{M})$, with a coordinate system denoted by $(x^1, \ldots x^{dim M})$, say, such that $\partial _{x^1} = \tilde{\nabla}\Omega|_{\tilde{U}}$. Then the set
\[
S := \{q \in \tilde{U} \, : \, x^1(p) = \Omega(q) =0 \}
\]
is a $C^{\infty}$ embedded codimension 1 submanifold (i.e., a hypersurface) of $\tilde{M}$ containing ${\cal J}^+_c \cap \tilde{U}$ and closed in $\tilde{U}$. Denote by $\tilde{U}_+$ the connected open subset of $\tilde{U}$ in which $\Omega >0$. Now, $M \cap \tilde{U}$ is obviously (non-empty and) contained in $\tilde{U}_+$, and if it were {\em properly} contained therein, then $(M\cup \tilde{U}_{+}, (1/\Omega^2) \tilde{g} |_{M\cup \tilde{U}_{+}})$ would be a non-trivial extension of $(M,g)$, contrary to our assumption. Therefore, $M \cap \tilde{U} = \tilde{U}_+$, and we then conclude that $\partial ^+ M \cap \tilde{U}\equiv S$, and hence that $S = {\cal J}^+_c \cap \tilde{U}$. This in turn establishes the first statement.

\end{proof}

We now turn to a generalized notion of null infinity. 
Let again a strongly causal spacetime $(M,g)$ with $c$-completion $\overline{M}$ and $c$-boundary $\partial M$ be given. Wherever confusion might arise between conformal and $c$-completions, we use lowercase letters $p,q,r, \ldots$ to denote points on $\overline{M}_i$, while pairs $(P,F)$ are used to denote elements of the $c$-completion $\overline{M}$. The chronological future/past of a subset $U$ on $\overline{M}$ (resp. $\overline{M}_i$) will be denoted by $I^{\pm}(U)$ (resp. $I_i^{\pm}(U)$).

\begin{definition}
\label{scri}
The {\em future null infinity} of $M$, denoted as ${\cal J}^+$, is the set of points $(P,F) \in \partial  M$ such that
\begin{enumerate}[label=(\roman*)]
\item $\exists$ a future-complete and future-regular null ray $\eta:[0,+\infty) \rightarrow M$ with $(P,F)$ as a future endpoint of $\eta$,

\item $\mbox{every future-inextedible null geodesic with future endpoint $(P,F)$ is fu\-tu\-re-com\-ple\-te.}$
\end{enumerate}
A time-dual analogous definition is immediate for the \emph{past null infinity} ${\cal J}^-$.
\end{definition}

\begin{remark}\label{rem:1}
  \em Condition (i) in the previous definition ensures, in particular, that for any $(P,F)\in \mathcal{J}^{+}$, $P\neq \emptyset$ (recall Proposition \ref{prop:strongfirst}).
\end{remark}
Definition \ref{scri} is meant to capture the standard physical notion of ``distant observers'' in the absence of a conformal boundary. This is clearly a {\em geometric} rather than just {\em conformal} notion, as shows up here in the geodesic completeness requirements. It seems that we do not miss the standard points at infinity at least when they are regular and the conformal extension is well-behaved. However, the example in Remark \ref{rmk1} (3) still shows that points at conformal infinity which are {\em not} regular do not need to be in ${\cal J}^+$ (modulo $\Psi$) as defined. Indeed, the point $u=v=\pi/2$ in that example is a point at conformal infinity which is not regular, and the corresponding element in $\partial  M$ is $(M,\emptyset)$ itself (the first component is a TIP in this case!). But $M$ is not the chronological past of any future-directed {\em null} ray (so clause (i) is not satisfied), although it is the past of (infinitely many) future-directed {\em timelike} rays, and clause (ii) in Definition \ref{scri} is trivially satisfied. The reason why we include the clause (i) in Definition \ref{scri} is because we are interested only in the {\em null} infinity here. Clause (ii), on the other hand, seems physically justified if our ``distant observers'' are not to ``see up close'' any (potentially ``fatal''!) ``naked singularity''.


\section{Generalized black holes}\label{bh2}
We now use the new notion of null infinity to define a generalized notion of black hole. There is, of course, little choice for this particular definition once we have defined ${\cal J}^+$.

\begin{definition}
\label{defbh} Let $(M,g)$ be a spacetime. We shall say that a point $p \in M$ is {\em visible from (the future null)  infinity} ${\cal J}^+$ if there exists a future-inextendible null geodesic $\alpha$ starting at $p$ \cambios{and with future endpoint} some $(P,F) \in {\cal J}^+$. We denote the set of all points of $M$ visible from infinity by $V_{\infty}$. If ${\cal J}^+ \neq \emptyset$ (or equivalently, if $V_{\infty} \neq \emptyset$), then the {\em black hole (region)} of $(M,g)$ and the {\em (future) event horizon} are, respectively,
\[
B^+ := M\setminus J^-(V_{\infty}),
\]
and
\[
H^+:= \partial B^+.
\]
\end{definition}
\noindent (We require that ${\cal J}^+ \neq \emptyset$ in this definition to avoid that $M = B^+$.)

\begin{remark}
\label{rmk5}
{\em Given any point $p \in V_{\infty}$, and a future-complete null geodesic $\alpha:[0,+\infty) \rightarrow M$ starting
at $p$ with endpoint at ${\cal J}^+$, we have $I^-(\alpha) = I^-(\alpha|_{[t,+\infty)})$ for any $t \in [0,+\infty)$, and therefore any point along $\alpha$ is also visible from infinity.}
\end{remark}

Some of the more basic properties of black holes in this context are summarized in the following proposition. We particularly call the reader's attention to BH4), which precisely clinches the notion of a black hole as a region where causal communication with infinity (i.e., ``distant observers'') is forbidden.
\begin{proposition}
\label{bhprops}
The following properties of the black hole region $B^+$ and the event horizon $H^+$ hold.
\begin{itemize}
\item[BH1)] $I^+(B^+) \subset int(B^+)$. In particular, the interior $int(B^+)$ of the black hole is a future set. Moreover, the event horizon is the common boundary between $int(B^+)$ and the past set $I^-(V_{\infty})$, i.e., $M$ is the disjoint union
\begin{equation}
\label{equality}
M = int(B^+) \dot{\cup} H^+ \dot{\cup}I^-(V_{\infty}).
\end{equation}
In particular, $H^+$ is an {\em achronal boundary} (and hence \cite{ONeillSemiRiemannianGeometryApplications1983,PenroseDifferentialTopology1972} a closed $C^0$ hypersurface) in $M$.
\item[BH2)] For any  $p \in H^+\cap J^-(V_{\infty})$ there exists a future-complete null geodesic ray $\eta \subset H^+$ starting at $p$ with future endpoint at ${\cal J}^+$.
\item[BH3)] $\overline{B^+} = M \setminus I^-(V_{\infty}) = M \setminus I^-({\cal J}^+)$.
\item[BH4)] $p \in M \setminus B^+$ if and only if there exists a future-directed causal curve in $M$ starting at $p$ with a future endpoint on ${\cal J}^+$.
\end{itemize}
\end{proposition}
\begin{proof}
  \textit{BH1):} Let $q,p \in M$, with $q \in B^+$ and $q \ll_g p$. Then,
\[
q \notin J^-(V_{\infty}) \Longrightarrow p \notin  J^-(V_{\infty}) \Longrightarrow p \in B^+,
\]
but this actually proves that $I^+(q) \subset B^+$, and since $I^+(q)$ is open and contains $p$, that $p \in int(B^+)$. Now, $\overline{B^+} = int(B^+)\dot{\cup} H^+$, and
\[
p \notin \overline{B^+} \Longrightarrow \exists U \ni p \mbox{ open with }U \cap B^+ =\emptyset \Longrightarrow U \subset J^-(V_{\infty}) \Longrightarrow p \in I^-(V_{\infty}).
\]
That is, $M \setminus \overline{B^+} \subset I^-(V_{\infty})$. These implications are easily reversed, so also $I^-(V_{\infty}) \subset M \setminus \overline{B^+}$. Therefore, $M\setminus\overline{B}^+=I^{-}(V_{\infty})$, whence (\ref{equality}) follows.

\smallskip

\textit{BH2):} Let $p \in H^+\cap J^-(V_{\infty})$, and let $\alpha$ be a future-directed causal curve segment starting at $p$ and ending at some point $q \in V_{\infty}$. By the definition of visible from infinity, there exists a future-complete null geodesic $\gamma$ starting at $q$ with a future endpoint $(P,F) \in {\cal J}^+$. Let $\eta$ be the juxtaposition of these two curves. This is a future-inextendible causal curve starting at $p$. Now, we claim that its image has to be globally achronal. Otherwise, we could find some $r \in \gamma$ with $p \ll_g r$. But $H^+ \equiv \partial I^-(V_{\infty})$ by BH1), and $r \in J^-(V_{\infty})$ (cf. Remark \ref{rmk5}), so this would mean that $p \in I^-(V_{\infty})\cap \partial I^-(V_{\infty})$, which is impossible. But since $\eta$ is globally achronal, it can be affinely reparametrized as null geodesic ray (which we still call $\eta$). Since $\gamma$ is future-complete, so is $\eta$. Clearly, $(P,F)$ is also an endpoint for $\eta$ which thus has a future endpoint on ${\cal J}^+$, as desired.

\smallskip

\textit{BH3):} The equality $\overline{B^+} = M \setminus I^-(V_{\infty})$ follows immediately from (\ref{equality}). Let $p \in I^-(V_{\infty})$. Then for some $q \in I^+(p)$, there exist a future-directed  null geodesic $\gamma$ starting at $q$ and with an endpoint $(P,F) \in {\cal J}^+$. Moreover, taking into account that $P\neq \emptyset$ (recall Remark \ref{rem:1}), item \ref{item:endpoint3} of Corollary \ref{cor:endpoints} implies that $P=I^{-}(\gamma)$. But this means $p \in P=I^-(\gamma)$ and hence, from \eqref{eq:7}
\[
(I^-(p),I^+(p))\ll (P,F) \in {\cal J}^+,
\]
that is, $(I^-(p),I^+(p)) \in I^-({\cal J}^+)\cap M$. Therefore  we have
\[
I^-(V_{\infty}) \subset I^-({\cal J}^+) \cap M.
\]
Conversely, let $p\equiv (I^-(p),I^+(p)) \ll (P,F) \in {\cal J}^+$. This means that $p \in P=I^-(\gamma)$ for some future-complete  null ray $\gamma$, and so $p\ll_g q$ for some $q \in \gamma$. But then $q$ is visible from infinity (by Remark \ref{rmk5}), so $p \in I^-(V_{\infty})$. This now implies the opposite inclusion \[ I^-({\cal J}^+) \cap M\subset I^-(V_{\infty}), \] and hence $M \setminus I^-(V_{\infty}) = M \setminus I^-({\cal J}^+)$.

\smallskip

\textit{BH4):} Let $\gamma:[0,A) \rightarrow M$ ($A \leq +\infty$) be a future-inextendible causal curve starting at $p =\gamma(0)$ and with a future endpoint $(P,F) \in {\cal J}^+$. Either $\gamma$ is a (necessarily future-complete) null geodesic, in which case $p$ is visible from infinity, or else $p \in I^-(\gamma)$. But then $(I^-(p),I^+(p)) \in I^-({\cal J}^+)$, and hence $p \in I^-(V_{\infty})$ from the proof of \textit{BH3)}. In any case $p \in J^-(V_{\infty})$ and so $p \notin B^+$.

Conversely, if $p \notin B^+$ then there exists a future-complete null geodesic $\gamma:[0,+\infty) \rightarrow M$ with $p \in J^-(\gamma(0))$
such that some $(P,F)\in {\cal J}^+$ is endpoint of $\gamma$. We can therefore juxtapose a future-directed causal curve from $p$ to $\gamma(0)$ with $\gamma$ to obtain a future-inextendible causal curve $\beta$ starting at $p$. Clearly, $(P,F)$ is also endpoint of $\beta$.
\end{proof}

At least in some important cases, one can draw a direct connection between the absence of a black hole region and nonspacelike geodesic completeness.

\begin{proposition}
\label{completeness1}
Suppose that $(M,g)$ is globally hyperbolic and future null complete, with non-compact Cauchy hypersurfaces. Then ${\cal J}^+ \neq \emptyset$ but $B^+ = \emptyset$.
\end{proposition}
\begin{proof}
Let $p \in M$. Assume first that $\partial I^+(p)$ is compact. In this case, a standard argument (see, e.g., the proof of Theorem 61, Ch.14, p. 437 of \cite{ONeillSemiRiemannianGeometryApplications1983}) implies that $\partial I^+(p)$ is homeomorphic to a given Cauchy hypersurface $S\subset M$, which is absurd. Therefore, $\partial I^+(p)$ is non-compact, and again by simple arguments using limit curves (cf. the proof of Prop. 8.18, Ch. 8, p. 289 of \cite{BeemGlobalLorentzianGeometry1996}) we conclude that there exists an inextendible future-directed null geodesic ray $\eta$ starting at $p$. Now, consider the TIP $P=I^{-}(\eta)$. By Remark \ref{r0}, we can pick $F \in \check{M}\cup\{\emptyset\}$ such that $(P,F) \in \overline{M}$ (actually, $F\equiv \emptyset$ will do). Since $P$ is a TIP, $(P,F) \in \partial M$. By Remark \ref{r}, global hyperbolicity guarantees that $\eta$ is future-regular, and Proposition \ref{prop:strongfirst} now implies that $(P,F)$ is a future endpoint of $\eta$. Finally, null geodesic completeness and the future-regularity of $\eta$ now imply that $(P,F) \in {\cal J}^+$. We conclude that $p \notin B^+$ by Proposition \ref{bhprops} BH4).
\end{proof}

If we assume some extra natural geometric conditions on $(M,g)$, we can show that the general conclusion of the previous proposition can be extended to strongly causal spacetimes.

\begin{theorem}
\label{completeness2}
Suppose that the strongly causal spacetime $(M^{n+1},g)$, with $n\geq 2$, satisfies the following conditions:
\begin{itemize}
\item[(a)] $(M,g)$ is timelike and null geodesically complete;
\item[(b)] $(M,g)$ satisfies the {\em timelike convergence condition}, i.e., $Ric(v,v)\geq 0$ for any timelike $v\in TM$;
\item[(c)] ${\cal J}^+ \neq \emptyset$;
  \item[(d)] $\overline{M}$ is strongly properly causal (recall this holds, in particular, if $(M,g)$ is causally continuous - cf. Prop. \ref{prop:causalcontinuity}.)
\end{itemize}
Then $B^+ = \emptyset$.
\end{theorem}
\begin{proof}
Suppose, by way of contradiction, that $(a)-(d)$ do hold, but $B^+ \neq \emptyset$. In this case, Proposition \ref{bhprops} BH1) also implies that ${\rm int} (B^+) \neq \emptyset$; therefore, we can pick $p \in {\rm int} (B^+)$. Now, if there exists a future-directed null ray $\eta$ starting at $p$, from $(a)$ and $(d)$ it must be both future-complete and future-regular. Therefore, arguing exactly as in the end of the proof of Proposition \ref{completeness1} with strong proper causality in lieu of global hyperbolicity, we would conclude that $p \notin B^+$, a contradiction. But if such a ray does not exist, then $p$ is a {\em future trapped set}, i.e., its future horismos $E^+(p):= J^+(p) \setminus I^+(p)$ is compact, again by the proof of  of \cite[Prop. 8.18, Ch. 8, p. 289]{BeemGlobalLorentzianGeometry1996}. Therefore, since $(M,g)$ is strongly causal, the latter result together with \cite[Theor. 8.13]{BeemGlobalLorentzianGeometry1996} imply that $(M,g)$ admits a causal line intersecting $E^+(p)$, $\gamma$ say.

We claim that $\gamma$ is actually a {\em timelike} line. Without loss of generality we may assume that $\gamma$ is future-directed. Now, Proposition \ref{bhprops} BH1) means that $int(B^+)$ is a future set, and hence $E^+(p)$ is contained in ${\rm int}(B^+)$. Therefore, we can pick some $x \in {\rm int}(B^+) \cap \gamma$. If $\gamma$ were null, then its portion to the future of $x$ would be a future-directed null geodesic ray, that can be considered complete and future-regular, that is, (again arguing exactly as in the proof of Proposition \ref{completeness1}) with future endpoint on ${\cal J}^+$, which again contradicts Proposition \ref{bhprops} BH4). Thus, $\gamma$ has to be timelike.

However, conditions $(a)$ and $(b)$, together with the existence of a (complete) timelike geodesic line mean that we can apply the Lorentzian Splitting Theorem (see, e.g., Ch. 14 of \cite{BeemGlobalLorentzianGeometry1996} and references therein) to conclude that $(M,g)$ is actually {\em isometric} to a product spacetime $(\mathbb{R}\times S,-dt^2\oplus h)$, where $(S,h)$ is a complete Riemannian manifold.

Now, Theorem 3.67 in Ref. \cite{BeemGlobalLorentzianGeometry1996} implies that $(M,g)$ is actually globally hyperbolic, with Cauchy hypersurfaces homeomorphic to $S$. Hence $S$ cannot be non-compact, for in that case Proposition \ref{completeness1} would mean that $B^+ \equiv \emptyset$, contrary to our assumption. We conclude that $S$ is compact.

But then no future-inextendible null geodesic can be achronal, i.e., a geodesic ray; so that in this case we would have ${\cal J}^+ \equiv \emptyset$, contradicting $(c)$. This final contradiction ends the proof.
\end{proof}

In order to get deeper results out of our black hole definition, we need to refine it. We shall need some additonal conditions ensuring that future null infinity is ``good enough''. The following definition is meant to encode this.
\begin{definition}
\label{ample}
The future null infinity ${\cal J}^+$ is said to be:
\begin{itemize}
\item[(A1)] {\em Ample} if for any compact set $C \subset M$, and for any connected component ${\cal J}^+_0$ of ${\cal J}^+$,
${\cal J}^+_0\cap (\overline{M} \setminus \widetilde{I^+ (C)})$ is a non-empty open set, where
   \begin{equation}
\widetilde{I^+(C)}:=\{(P,F)\in \overline{M}: I^-(x)\subset P \mbox{ for some $x\in C$}\}.\label{eq:2}
\end{equation}
(Roughly speaking, no connected component of future null infinity can be entirely contained in the future of a compact set.)


\item[(A2)] \emph{Past-complete} if given $(P,F)\in \mathcal{J}^+$, any $(P',F')\in\partial M$ with $P'=I^-(\eta)$, being $\eta$ a future-directed inextendible null geodesic generator of $\partial P$, also belong to $\mathcal{J}^+$.


\end{itemize}
We will say that the future null infinity ${\cal J}^+$ is \emph{regular} if it is both ample and past-complete.
\end{definition}
Condition (A1) in the previous definition means that the future of compact sets cannot encopass the whole future null infinity, and (A2) that ${\cal J}^+$ contains any point on the future boundary which may lie in its past. Together, they mean that the future null infinity is "big" in a precise sense. These assumptions are not really restrictive when some classical examples of physical interest are considered. In fact, on the one hand, the condition that the {\em conformal} $\mathcal{J}^+$ escapes from the future of any compact set holds, for instance, in many standard solutions of Einstein field equation, with vanishing cosmological constant, admitting a conformal completion for which conformal null infinity $\mathcal{J}^+$ is a null hypersurface in the extended spacetime having past-complete null generators, such as those in the Kerr-Newman family with suitable parameters. It also holds in some solutions to the Einstein fields equation with a negative cosmological constant, such as the Schwarzschild-Anti de Sitter spacetime, for example. On the other hand, it fails on, say, the Schwarzschild-de Sitter solutions, and will tend to fail, more generally, on globally hyperbolic spacetimes with compact Cauchy hypersurfaces. In any case, the seemingly technical assertion on the open character of the intersection considered in (A1) holds in very general situations, as becomes apparent from the following result:



\begin{proposition}
\label{lema:auxiliar}
If $\hat{M}$ is Hausdorff, then $\widetilde{I^+(C)}$ is a closed set.
\end{proposition}
\begin{proof}
Let $\left\{(P_n,F_n)  \right\}_{n}\subset \widetilde{I^+(C)}$ be a sequence and consider $(P,F)\in L(\left\{ (P_{n},F_n) \right\}_{n})$. Our aim is to prove that $(P,F)\in \widetilde{I^{+}(C)}$. For any $n$, let $x_{n}\in C$ be a point so $I^-(x_n)\subset P_{n}$. Observe that, as $C$ is compact, we can assume (up to a subsequence) that $\{x_{n}\}_{n}$ converges in $C$ to a point, say $x^{*}\in C$.

From such a convergence it follows (see \cite[Remark 3.17]{Floresfinaldefinitioncausal2011}) that $I^{-}(x^{*})\subset \mathrm{LI} (\left\{ I^-(x_{n}) \right\}_{n})$, and so, that $I^{-}(x^{*})\subset \mathrm{LI} (\left\{ P_n\right\}_{n})$. From standard arguments involving Zorn's Lemma, and up to a subsequence, we can ensure the existence of a IP $\overline{P}$ so $I^{-}(x^{*})\subset \overline{P}$ and $\overline{P}$ is maximal on $\mathrm{LS} (\left\{ P_{n} \right\}_{n})$, i.e., $\overline{P}\in \hat{L}(\left\{ P_{n} \right\}_{n})$. As $\hat{M}$ is Hausdorff and $P\in \hat{L}(\left\{ P_n \right\}_{n})$, $P=\overline{P}\supset I^{-}(x^{*})$, proving that $(P,F)\in \widetilde{I^+(C)}$.
\end{proof}

Condition (A2) in Definition \ref{ample} is comparatively more restrictive than (A1). As we will see on the Appendix, this condition seems unavoidable if we are to obtain the following theorem, insofar as imposing stronger requirements on the causality of the underlying spacetime will not help one to evade it.

\begin{theorem}
\label{main}
Assume that ${\cal J}^+$ is regular and let $C \subset M$ be an {\em achronal} compact set. If $C$ is not entirely contained in $B^+$, then there exists a future-complete null $C$-ray $\eta: [0,+\infty) \rightarrow M$ \cambios{with endpoint} $(P,F) \in {\cal J}^+$.
\end{theorem}
\begin{proof}
First, assume that $C$ is not entirely contained in $\overline{B^+}$. Then, BH3) in Proposition \ref{bhprops} implies that ${\cal J}^+\cap I^+ (C)\neq \emptyset$. Let $x \in {\cal J}^+\cap I^+ (C)$, and pick a TIP $P_0$ such that $x \in P_0$ and $(P_0,F_0)\in {\cal J}^+$. Denote by ${\cal J}^+_0$ the connected component of ${\cal J}^+$ containing $(P_0,F_0)$.

\cambios{Since ${\cal J}^+$ is ample, we have
\[
{\cal J}_0^+\cap(\overline{M}\setminus I^+(C))\supset {\cal J}_0^+\cap(\overline{M}\setminus\widetilde{I^+(C)}) \neq \emptyset.
\]
So, taking into account that ${\cal J}^+_0$ is connected, and $(P_0,F_0)\in {\cal J}^+_0\cap I^+ (C)\neq\emptyset$. Thus,
\[
\partial_{{\cal J}^+_0}({\cal J}^+_0\cap I^+ (C))\neq\emptyset.
\]
Take some $(P,F)\in \partial_{{\cal J}^+_0}({\cal J}^+_0\cap I^+ (C))$. Since $I^+ (C)$ is open, it follows that $(P,F) \notin I^+ (C)$ and in particular $C \cap P = \emptyset$.}

We now claim that $I^{-}(x_0)\subset P$ for some $x_0\in C$. Suppose, by way of contradiction, that this is false. Then, it follows that $(P,F)\in U=\mathcal{J}_{0}^{+}\cap\left(\overline{M}\setminus\widetilde{I^{+}(C)}\right)$, which is an open set from the ample condition. As $I^{+}(C)\subset \widetilde{I^{+}(C)}$, then $I^{+}(C)\cap U=\emptyset$. Hence, the point $(P,F)$ belongs to an open set with empty intersection with $I^{+}(C)$, in contradiction with $(P,F)\in\partial_{{\cal J}^+_0}(\mathcal{J}^+_{0}\cap I^+ (C))$. This establishes the claim.

Therefore, $x_0 \in \overline{I^{-}(x_0)}\subset \overline{P}$, but $x_0\in C\setminus P$. We conclude that $x_0 \in\partial_{M} P$. Since $(P,F) \in \partial M$, in particular $P$ is a TIP. Its boundary in $M$ is thus a union of future-inextendible null geodesics, so that we can take $\eta$ the future-inextendible null geodesic generator of $\partial_{M} P$ starting at $x_0$. Note that since $I^-(\eta)$ is also a TIP, we conclude that $(P',F') \in {\cal J}^+$ with $P'=I^-(\eta)$ by clause (A2) in Definition \ref{ample}, and in particular $\eta$ must be future-complete by clause (ii) in Definition \ref{scri}.

We wish to show that $\eta$ is a $C$-ray. By construction, $\eta \subset \overline{I^+(C)}$, but since $\partial_{M} P \cap I^+(C) = \emptyset$, $\eta \subset \partial_{M} I^+(C)$. Moreover, since $C$ is achronal, $C \subset \partial_{M} I^+(C)$. Finally, since  $\partial_{M} I^+(C)$ is an achronal set, any causal curve segment connecting a point of $C$ with a point $x$ (say) along $\eta$ must be a reparametrization of a null geodesic, and in particular has zero Lorentzian arc-length. This means that the initial segment of $\eta$ between $C$ and $x$ is maximal. Hence, $\eta$ is a $C$-ray as claimed.

Now, assume that $C$ is contained in $\overline{B^+}$, but not in $B^+$. We can then pick a point $p \in H^+\cap C\cap J^-(V_{\infty})$. By the item BH2) of Proposition \ref{bhprops}, there exists a future-complete null geodesic ray, which we again denote by $\eta$, starting at $p$ and with future endpoint on ${\cal J}^+$. We only need to check it is again a $C$-ray. But if this were not the case, then there would exist some $q \in C$ and some point $r$ along $\eta$ with $q \ll_g r$. But since $r$ is visible from infinity (cf. Remark \ref{rmk5}), this would mean that $q \in I^-(V_{\infty}) \equiv M\setminus \overline{B^+}$ (cf. BH3) in Proposition \ref{bhprops}), a contradiction.
\end{proof}

We can now use this theorem to prove a classic result in the theory of black holes \cite{BeemGlobalLorentzianGeometry1996,HawkingLargeScaleStructure1975,ONeillSemiRiemannianGeometryApplications1983}  for this extended context. Specifically, we wish to show that {\em any closed trapped surface stays inside the black hole region}, or, in other words, ``hidden from distant observers'' by the event horizon.

Recall that a {\em closed (future) trapped surface} in $(M,g)$ is a smooth, codimension 2, spacelike, achronal, compact submanifold $S \subset M$ without boundary whose mean curvature vector field is everywhere past-directed timelike \cite{ONeillSemiRiemannianGeometryApplications1983}. (The presence of this geometric object was introduced by Penrose \cite{PenroseGravitationalCollapse1965} as the mathematical surrogate for ``a point of no return'' in gravitational collapse.)  Also, recall that the {\em null convergence condition} holds in $(M,g)$ when $Ric(v,v) \geq 0$ for each null $v \in TM$. For solutions of the Einstein field equation of General Relativity (with or without a cosmological constant), this condition is implied by all the standard ``energy conditions'' on the stress-energy tensor, such as the dominant energy condition or the weak energy condition \cite{HawkingLargeScaleStructure1975,WaldGeneralRelativity1984}.

\begin{corollary}
\label{trappedcor}
Assume that ${\cal J}^+$ is regular and that the null convergence condition holds in $(M,g)$. If $S \subset M$ is a closed trapped surface, then $S \subset B^+$.
\end{corollary}
\begin{proof}
Suppose, to the contrary, that $S$ is not contained in $B^+$. Since $S$ is in particular achronal and compact by definition, Theorem \ref{main} implies the existence of some future-complete null $S$-ray $\eta: [0,+\infty) \rightarrow M$ with future endpoint $(P,F)\in {\cal J}^+$. But $\eta$ would be then a future-complete normal null geodesic ray starting at $S$ and without focal points, which contradicts standard results for closed trapped surfaces in spacetimes where the null converge condition holds (see, e.g. \cite{BeemGlobalLorentzianGeometry1996,HawkingLargeScaleStructure1975,ONeillSemiRiemannianGeometryApplications1983}).
\end{proof}

Finally, if ${\cal J}^+$ is regular, we can strengthen the item BH4) of Proposition \ref{bhprops} as follows.
\begin{corollary}
  Assume that ${\cal J}^+$ is regular. Then $p \in M \setminus B^+$ if and only if $p$ is visible from infinity.
\end{corollary}
\begin{proof}
Immediate from Theorem \ref{main} by taking $C= \{p\}$.
\end{proof}


\section{Application to generalized plane waves}\label{ppwaves}

We are going to study in this section the null infinity of, and prove the absence of black holes in, the class of {\em generalized plane waves}. This family of spacetimes has been intensely studied in the literature, and is especially relevant for us here because conformal boundaries are not always available therein, but $c$-boundaries do, since they are strongly causal under mild assumptions \cite[Section 3]{0264-9381-20-11-322}. Hence, definitions of null infinity and black holes based on the latter become a natural alternative. In that case, the classical notions of null infinity and/or black holes do not make sense. (For a different perspective on the same issue, see \cite{SenovillaNoBHinPPWaves2003}.) Let us begin by reviewing some generalities about these spacetimes.

A {\em generalized plane wave} is any spacetime $(M^{n+1},g)$ of the form
\begin{equation}\label{pfw}
M=M_0 \times \mathbb{R}^2,\qquad
g(\cdot,\cdot) = g_0(\cdot,\cdot) + 2dudv + H(x,u)du^2,
\end{equation}
where $g_0$ is a smooth Riemannian metric on the ($n-1$)-dimensional manifold $M_0$, the variables $(v,u)$ are the standard Cartesian coordinates of $\mathbb{R}^2$, and $H:M_0 \times \mathbb{R} \rightarrow \mathbb{R}$ is some smooth real function.

The vector field $\partial_v$ is parallel (i.e.,
covariantly constant) and null, and the time-orientation will
always be chosen which makes it past-directed. In particular, its integral curves are past-directed null geodesics.

Consider any future-directed
causal curve segment $t \in [a,b] \mapsto \gamma(s)=(x(s),v(s),u(s))$ in $(M,g)$. Since the gradient $\nabla \,u \equiv \partial_v$ is past-directed null, we have
\begin{equation}
\label{cute}
\dot{(u\circ \gamma)}(s) = g(\nabla \, u (\gamma(s)), \dot \gamma(s)) = g(\partial_v|_{\gamma(s)},\dot \gamma(s)) = \dot u(s) \geq 0,
\end{equation}
and the inequality is strict whenever $\dot \gamma(s)$ is timelike. Integrating Eq. (\ref{cute}), we get
\[
u(b) - u(a)\geq 0,
\]
with strict inequality unless $\gamma$ is a null pregeodesic without conjugate points contained in a $u\equiv\hbox{constant}$ hypersurface. If the latter is not the case, then $\gamma$ leaves any such hypersurface where it starts, and in particular $\gamma$ cannot be a closed curve.

This calculation reveals, in particular, that each $u\equiv\hbox{constant}$ hypersurface is null and achronal and $(M,g)$ does not admits any closed timelike curves, i.e., it is chronological. The null geodesic generators of any such hypersurface coincide with the maximal integral curves of $\partial_{v}$ therein, and achronality implies that any null geodesic inside these hypersurfaces coincide with such generators.

Finally, a closer analysis of the geodesic equations easily shows that any null geodesic inside a $u\equiv\hbox{constant}$ hypersurface is {\em injective} and {\em complete}. We conclude that {\em every generalized plane wave is causal}; moreover, the vector $\partial_v$ is complete and all its maximal integral curves are null geodesic lines.

Fixing some local coordinates $x_1, \dots, x_{n-1}$ for the Riemannian
part $M_0$,
%
%
%
%
the three geodesic equations for
a curve $\gamma(s)= (x(s), v(s), u(s))$, $s\in (a,b)$, can be solved in
the following three steps  \cite[Proposition 3.1]{CandelaGeneralPlaneFronted2003}:
\begin{enumerate}
\item[(a)] $u(s)$ is any affine function, $u(s) = u_0 + s \Delta
u$, for some constant $\Delta u\in {\mathbb R}$;

\item[(b)] $x = x(s)$ is a solution on $M_0$ of
\[
D_s\dot x = - {\rm grad}_x V_{\Delta}(x(s),s) \quad \mbox{for all
$s \in \ (a,b)$,}
\]
where $D_s$ denotes the covariant derivative and $V_{\Delta}$ is
defined as:
\[
V_{\Delta}(x,s) := -\ \frac{(\Delta u)^2}{2}\ H(x, u_0 + s \Delta
u);
\]

\item[(c)] finally, with a fixed $v_{0}$ and an $s_0\in (a,b)$,
$v(s)$ can be computed from
\[
v(s) = v_0 + \frac{1}{2 \Delta u} \int_{s_0}^s \left( E_{\gamma} -
g_0(\dot x(\sigma), \dot x(\sigma)) + 2
V_{\Delta}(x(\sigma), \sigma)\right) d\sigma,
\]
where $E_{\gamma}=g(\dot{\gamma}(s), \dot{\gamma}(s))$ is a
constant (if $\Delta u = 0$ then $v = v(s)$ is affine).
\end{enumerate}
In particular, if we fix some $(\overline{x},\overline{u})\in M_0\times {\mathbb R}$, then the curve
\[
\gamma_{\overline{x},\overline{u}}(s)=(\overline{x},-s,\overline{u})\quad\hbox{(resp. $\beta_{\overline{x},\overline{u}}(s)=(\overline{x},s,\overline{u}))$,}\quad s\in {\mathbb R},
\]
is a future-directed (resp. past-directed) null geodesic line.

The $c$-boundary for generalized plane waves has been systematically studied in \cite{Florescausalboundarywavetype2008}, where it was shown that its structure strongly depends on the growth of $H$ at infinity. Some of these asymptotic behaviours, which will be also relevant in our discussion, are the following:
\begin{definition} A function ${\cal F}:M_0\times{\mathbb R}\rightarrow {\mathbb R}$ is said to be:
\begin{itemize}
%
\item[(1)] {\em at most quadratic}, if there exist $\hat{x}\in M_0$ and positive continuous functions $R_0(u), R_1(u) > 0$ such that
\[
{\cal F}(x,u)\leq R_1(u)d(x,\hat{x})^2 + R_0(u),\quad\forall (x,u)\in M_0\times {\mathbb R};
\]
\item[(2)] {\em $\lambda$-asymptotically quadratic, with} $\lambda>0$, if $M_0$ is non-compact and there
exist $\hat{x}\in M_0$, continuous functions $R_0(u), R_1(u)>0$ and a constant $R_0^-\in {\mathbb R}$ such that:
\[
\frac{\lambda^2d(x,\hat{x})^2+R_0^-}{u^2+1}\leq {\cal F}(x,u)\leq R_1(u)d(x,\hat{x})^2 + R_0(u),\quad\forall (x,u)\in M_0\times {\mathbb R}.
\]
\end{itemize}
\end{definition}
Let us recall now some noteworthy statements about the structure of the $c$-boundary for generalized plane waves (see \cite[Theorems 7.9 and 8.2]{Florescausalboundarywavetype2008}; note that $F$ in that reference corresponds with $-H$ here). From now on, we will assume that $(M_0,g_0)$ {\em is geodesically complete}:
\begin{theorem}\label{t} Let $(M,g)$ be a generalized plane wave as in (\ref{pfw}). Then, the following assertions hold.
\begin{itemize}
\item[(i)] If $|H|$ is at most quadratic, then the future (resp. past) $c$-boundary of $(M,g)$ {\em contains} a
copy $L^+$ (resp. $L^-$) of ${\mathbb R}$ plus the ideal point $i^+$ (resp. $i^-$)\footnote{Here, by $i^{+}$ and $i^{-}$ we are denoting the pairs $(M,\emptyset)$ and $(\emptyset,M)$, and the entire manifold $M$ is a terminal set.}. Every ideal point $\overline{u}\in L^{+}$ (resp. $\overline{u}'\in L^{-}$) can be identified with the IP $I^{-}(\gamma_{\overline{x},\overline{u}})$ for any $\overline{x}\in M_0$ (resp. the IF $I^{+}(\gamma_{\overline{x}',\overline{u}'})$, for any $\overline{x}'\in M_0$).

    The (total) $c$-boundary of $(M,g)$ {\em contains} a subset which can be identified with the quotient space $$((L^{+}\cup \{i^{+}\})\cup (L^{-}\cup \{i^{-}\}))/R,$$ where $R$ is the equivalence relation obtained by symmetrizing the following relation:\footnote{From \cite[Remark 7.10]{Florescausalboundarywavetype2008}, a pair $(\overline{x},\overline{u})$ cannot be related with more than one pair $(\overline{x}',\overline{u}')$, and viceversa.}
    \[
    (\overline{x},\overline{u})R(\overline{x}',\overline{u}')\quad\Leftrightarrow\quad (I^{-}(\gamma_{\overline{x},\overline{u}}),I^{+}(\gamma_{\overline{x}',\overline{u}'}))\in\overline{M}.
    \]

\item[(ii)] If $-H$ is $\lambda$-asymptotically quadratic for some
$\lambda>1/2$, then the future, past and total $c$-boundaries not only contains the structures described in (i), but necessarily {\em coincide} with them.



\end{itemize}
\end{theorem}


Now, we are ready to obtain the (future) null infinity for these spacetimes according to Definition \ref{scri}. To simplify the study, we will restrict our attention to {\em geodesically complete} generalized plane waves. This property is guaranteed, for instance, if $H(x,u)\equiv H(x)$ is at most quadratic (see \cite[Corollary 3.4]{CandelaGeneralPlaneFronted2003}), but there are pretty more situations were it holds. We will also assume that the null rays $\gamma_{\overline{x},\overline{u}}$ are future-regular.
\begin{corollary} Let $(M,g)$ be a geodesically complete generalized plane wave whose null rays $\gamma_{\overline{x},\overline{u}}$ are future-regular. Then, the following assertions hold:
\begin{itemize}
\item[(i)] If $|H|$ is at most quadratic, then ${\cal J}^+$ {\em contains} all the pairs of the form $(P,F)$, with $P\neq\emptyset$, which appear in case (i) of Theorem \ref{t}.
\item[(ii)] If $-H$ is $\lambda$-asymptotically quadratic for some
$\lambda>1/2$, then ${\cal J}^+$ not only contains, but also {\em coincides} with the structure described in previous case (i).
\end{itemize}
\end{corollary}

\begin{proof}
(i) Recall that $\gamma_{\overline{x},\overline{u}}=(\overline{x},-s,\overline{u})$, is a future-directed null geodesic line for every $(\overline{x},\overline{u})\in M_0\times {\mathbb R}$. Since $\uparrow \gamma_{\overline{x},\overline{u}}=\uparrow I^-(\gamma_{\overline{x},\overline{u}})$, the curve $\gamma_{\overline{x},\overline{u}}$ has a future endpoint of the form $(I^-(\gamma_{\overline{x},\overline{u}}),F)$, where either $F=I^+(\gamma_{\overline{x}',\overline{u}'})$ or $F=\emptyset$.
Moreover, since $(M,g)$ is assumed to be geodesically complete, any other inextendible future-directed null geodesic $\alpha$ with future endpoint $(I^-(\gamma_{\overline{x},\overline{u}}),F)$ is complete. Hence, $(I^{-}(\gamma_{\overline{x},\overline{u}}),F)$ belongs to ${\cal J}^+$ for every $(\overline{x},\overline{u})\in M_0\times {\mathbb R}$.

The argument for the case (ii) is totally analogous.
\end{proof}

\noindent As a direct consequence we can now deduce the absence of BH for these spacetimes.
\begin{corollary}\label{cc}
 If $(M,g)$ is a geodesically complete generalized plane wave whose null rays $\gamma_{\overline{x},\overline{u}}$ are future-regular, then it does not contain black holes. In particular, causally continuous, geodesically complete generalized plane waves have no black holes.
\end{corollary}

\begin{proof}
Given an arbitrary event $p_0=(x_0,v_0,u_0)\in M$, it suffices to show that $p_0\in I^-({\cal J}^+)$. From the proof of the previous theorem, $(I^{-}(\gamma_{\overline{x},\overline{u}}),F)\in {\cal J}^+$ for all $(\overline{x},\overline{u})\in M_0\times {\mathbb R}$. Moreover, $p_0=(x_0,v_0,u_0)\in I^-((I^{-}(\gamma_{\overline{x},\overline{u}}),F))$ if and only if $u_0<\overline{u}$. Hence, $p_0\in I^-({\cal J}^+)$ whenever $u_0<\overline{u}$. In conclusion, $M\subset I^-({\cal J}^+)$, and thus, $B^+\subset M\setminus I^-({\cal J}^+)=\emptyset$.

For the last assertion, just recall that any causally continuous spacetime is strongly properly causal (Proposition \ref{prop:causalcontinuity}), and thus, the null rays $\gamma_{\overline{x},\overline{u}}$ are future-regular (Definition \ref{def:sproperlycausal}).
\end{proof}

\begin{remark} {\em In \cite[Thm. 6.9]{EHRLICH_1992} (see also \cite{minguzzi12:causal_kam}) the authors provide mild conditions under which a plane wave is causally continuous. So, according to Corollary \ref{cc}, these conditions joined to the previously cited ones for geodesic completeness ensure that a plane wave has no black holes.
}
\end{remark}

\section*{Appendix}
\label{sec:some-examples}
In Definitions \ref{scri} and \ref{ample} we have included some clauses that, even though natural when interpreted from the classical conformal approach viewpoint, might conceivably be discarded in favor of less technical-looking ones. The following construction shows that this is not the case if the thesis of Theorem \ref{main} is to be preserved. In fact, we will display a globally hyperbolic spacetime which is not past-complete and where (unsurprisingly) Theorem \ref{main} fails. The example also shows that the Theorem \ref{main} is also false if past-completeness is assumed but condition (ii) on Definition \ref{scri} is removed. This will show in particular that, even if with strong causality requirements on the spacetime, one cannot expect to conserve Theorem \ref{main} if the past-completeness condition is removed from the notion of regularity in Definition \ref{ample}.

\begin{note}
\emph{The construction considered below is given by making simple modifications of Minkowski spacetime. It is not difficult to realize that the future c-boundaries of the resulting spacetimes are always Hausdorff. In particular, from Proposition \ref{lema:auxiliar}, the set $\widetilde{I^+(C)}$ will be closed for any compact set $C$ in the spacetime.}
\end{note}





Consider the $3$-dimensional Minkowski spacetime
\[
\mathbb{L}^3 = (\mathbb{R}^3, dx^2+dy^2-dt^{2}).
\]
Isometrically compactify its (cartesian) $x$-coordinate as a circle $\mathbb{S}^1$ in order to obtain a new (flat) Lorentzian manifold $\tilde{M}$ whose spatial sections are $2$-dimensional cylinders. Note that the future-complete null lines of $\tilde{M}$ consist of straight lines of $\tilde{M}$ with spatial component parallel to the $y$-axis. In fact, the spatial component of any other inextendible lightlike geodesic $\sigma$ in $\tilde{M}$ will ``waste time'' spinning around the spatial cylinder, and eventually, two points of $\sigma$ will become chronologically related by some timelike curve $c$ (see Figure \ref{fig:1}).

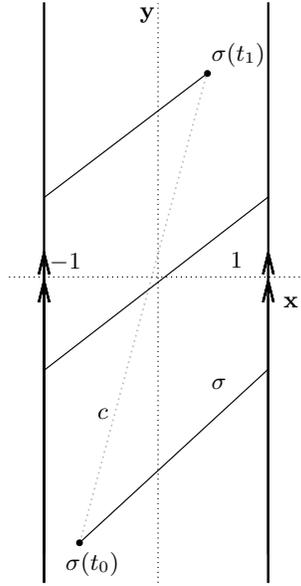
\begin{figure}
\centering

  \setlength{\unitlength}{1bp}%
  \begin{picture}(99.80, 230.32)(0,0)
  \put(-13,0){\includegraphics{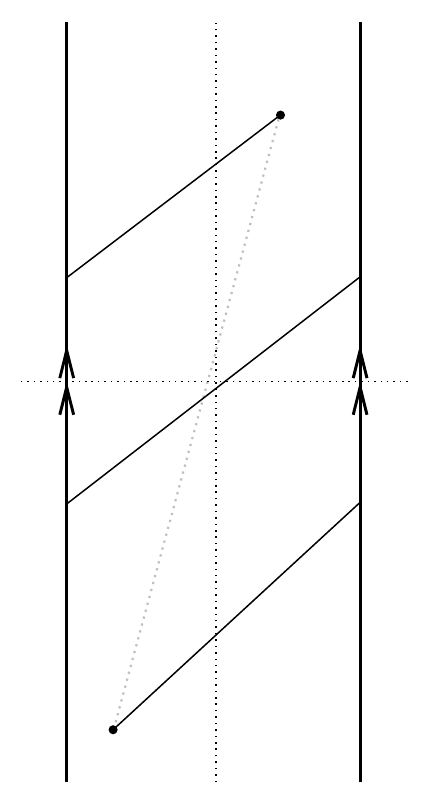}}
  \put(76.52,124.51){\rotatebox{0.00}{\fontsize{8.54}{10.24}\selectfont \smash{\makebox[0pt][l]{$1$}}}}
   \put(96.52,109.76){\fontsize{8.54}{10.24}\selectfont $\mathbf{x}$}
  \put(8.37,124.28){\rotatebox{0.00}{\fontsize{8.54}{10.24}\selectfont \smash{\makebox[0pt][l]{$-1$}}}}
  \put(69.36,78.76){\fontsize{8.54}{10.24}\selectfont $\sigma$}
  \put(14.36,10.76){\fontsize{8.54}{10.24}\selectfont $\sigma(t_0)$}
  \put(26.36,67.76){\fontsize{8.54}{10.24}\selectfont $c$}
  \put(69.36,202.76){\fontsize{8.54}{10.24}\selectfont $\sigma(t_1)$}
   \put(42.36,218.76){\fontsize{8.54}{10.24}\selectfont $\mathbf{y}$}
  \end{picture}%

  \caption{\label{fig:1} Representation of the spatial component of the spacetime $\tilde{M}$. If the spatial component of a lightlike geodesic $\sigma$ in $\tilde{M}$ is not parallel to the $y$-axis, then two points of $\sigma$ will eventually become chronologically related in $\tilde{M}$ by some timelike curve $c$.}

\end{figure}

Next, consider the compact set
\[
C=\{(0,y,0)\, : \, -1\leq y \leq 1\},
\]
and define

  \[
\Sigma:= \{(0,y,1)\, : \, y \in \mathbb{R}\}\cap J^+(C,\tilde{M}).
    \]
Consider the Lorentzian manifold $(M,g)$, where $M:=\tilde{M}\setminus J^+(\Sigma,\tilde{M})$ and $g$ is the induced metric on $M$ from $\tilde{M}$ (see figures  \ref{fig:3}  and \ref{fig:2} for illustrations of the projections of $M$ onto the $y=0$ and $x=0$ planes, resp.).

  The $c$-boundary of $(M,g)$ is the disjoint union of the future and past c-boundaries, each being formed by pairs with an empty $F$ or $P$ component, resp. In particular, $(M,g)$ is globally hyperbolic according to \cite[Theorem 3.29]{Floresfinaldefinitioncausal2011}. Moreover, the pairs $(P,\emptyset)$ of the future c-boundary can be separated in two classes: (a) those pairs defined by inextensible timelike curves with divergent component $y$, and (b) those pairs identifiable with the points of the (topological) boundary of $J^+(\Sigma,\tilde{M})$. It readily follows that the former points are in $\mathcal{J}^{+}$, since there exist complete null rays defining the corresponding TIPs (here we can proceed just as in Minkowski spacetime). However, the latter pairs belong to $\partial M\setminus \mathcal{J}^+$, since the null rays defining such TIPs are incomplete.  In any case, the following properties hold:

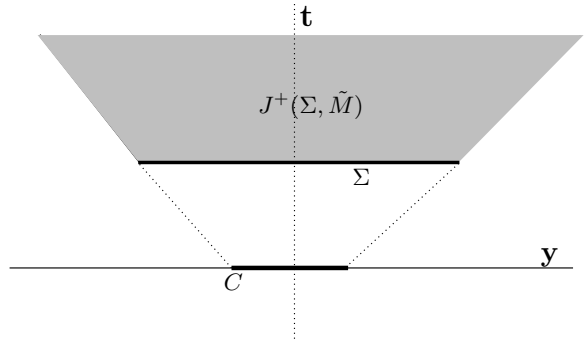
\begin{figure}
\centering
  \setlength{\unitlength}{1bp}%
  \begin{picture}(227.21, 140.19)(0,0)
    \put(0,0){\includegraphics{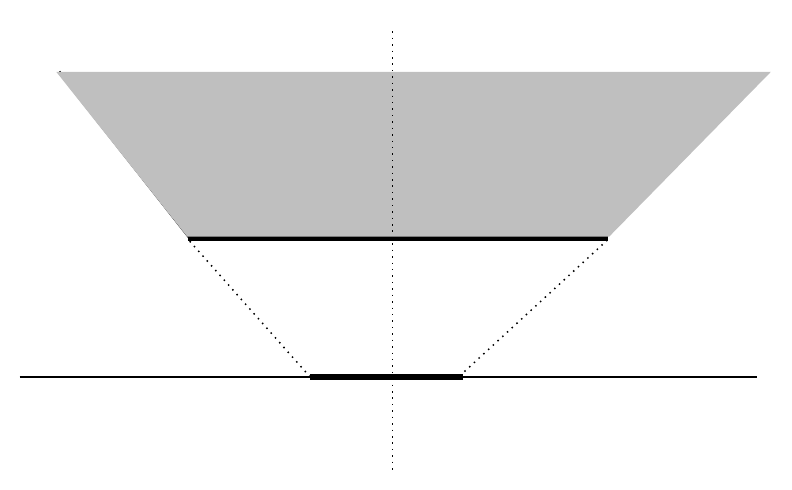}}
  \put(205.82,35.23){\fontsize{11.16}{13.00}\selectfont $\mathbf{y}$}
  \put(86.47,24.23){\fontsize{9.16}{11.00}\selectfont $C$}
  \put(134.95,63.63){\fontsize{9.16}{11.00}\selectfont $\Sigma$}
  \put(99.32,90.85){\fontsize{9.16}{11.00}\selectfont $J^+(\Sigma,\tilde{M})$}
  \put(115.03,123.78){\fontsize{11.75}{13.50}\selectfont $\mathbf{t}$}
  \end{picture}%
  \caption{\label{fig:3} Representation of the intersection of $(M,g)$ with the plane $x=0$. This is a standard plane with both the grey area and the set $\Sigma$ removed. All future-complete null rays in $M$ departing from points of the form $(t,0,y)$ are contained in this plane. In particular, any null ray departing from $C$ intersects $\Sigma$, and thus, there are no future-complete null $C$-rays.}
\end{figure}

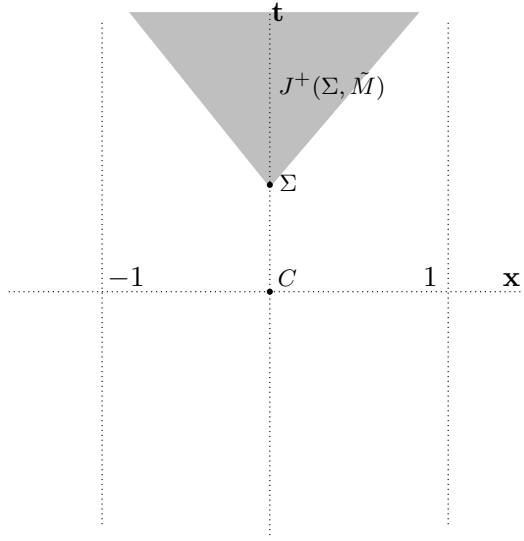
\begin{figure}
\centering
  \setlength{\unitlength}{1bp}%
  \begin{picture}(216.58, 213.83)(0,0)
  \put(0,0){\includegraphics{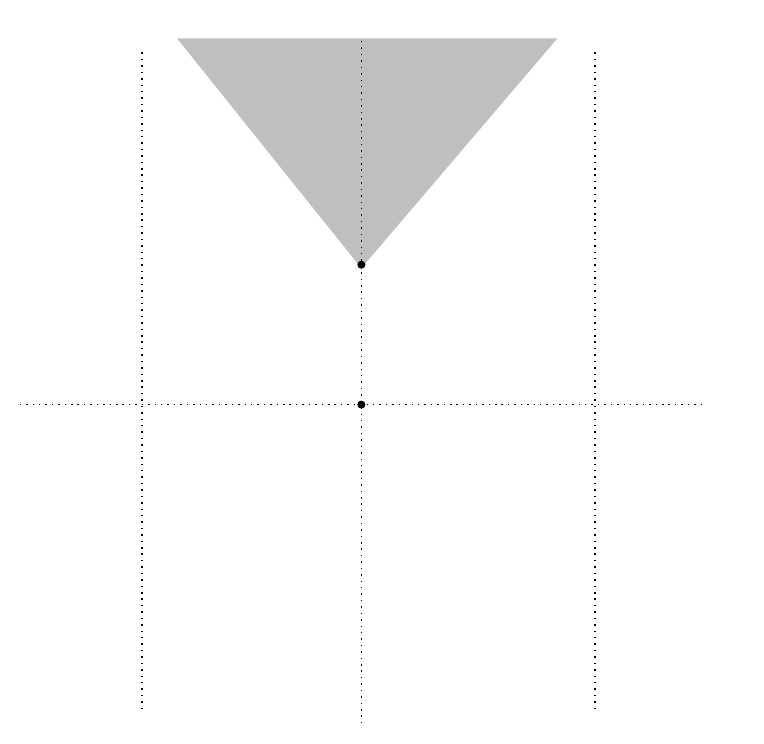}}
  \put(104.74,199.52){\fontsize{11.07}{13.28}\selectfont $\mathbf{t}$}
  \put(191.92,99.81){\fontsize{11.07}{13.28}\selectfont $\mathbf{x}$}
  \put(161.92,99.81){\fontsize{11.07}{13.28}\selectfont $1$}
   \put(42.92,99.81){\fontsize{11.07}{13.28}\selectfont $-1$}
  \put(107.35,172.11){\fontsize{9.38}{9.86}\selectfont $J^+(\Sigma,\tilde{M})$}
  \put(107.71,135.42){\fontsize{9.38}{9.86}\selectfont $\Sigma$}
  \put(107.07,99.46){\fontsize{9.38}{9.86}\selectfont $C$}
  \end{picture}%
  \caption{\label{fig:2} Representation of the intersection of $(M,g)$ with the plane $y=0$. Here, the projections of $\Sigma$ and $C$ are points, and the lines $x=-1$, $x=1$ of this plane are identified.}
\end{figure}

\begin{itemize}
\item ${\cal J}^+$ is ample: Let $K\subset M$ be any compact set. In order to show the existence of points in ${\cal J}^+$ not contained in the closed set $\widetilde{I^+(K)}$, let $(x_0,y_0,t_0)$ be a point in $M$ such that both $K$ and $J^+(\Sigma)$ are contained in $I^+(x_0,y_0,t_0)$. For any $\epsilon>0$ the null line $\gamma(t)=(x_0,y_0+t,t_0-\epsilon+t)$ is contained in $M$ (since it does not intersect $J^+(\Sigma)$), and defines a point $(P,\emptyset)\in {\cal J}^+$ with $(x_0,y_0,t_{0})\not \in P=I^-(\gamma)$. Moreover, the pair $(P,\emptyset)$ is not contained in $\widetilde{I^{+}(K)}$. Indeed, otherwise there would exist
some $q\in K$ such that $I^{-}(q)\subset P$. But by construction,
$q\in I^{+}(x_0,y_0,t_{0})$, and thus, $(x_0,y_0,t_{0})\in P$, in contradiction with the properties of $\gamma$.



\item $\mathcal{J}^{+}$ is not past-complete: Consider the null line $\sigma$ in $M$ given by $\sigma(t)=(1/2,1+t,t)$, which defines a pair $(P,\emptyset)\in \mathcal{J}^+$ with $P=I^-(\sigma)$. By construction, $\partial P$ is a plane containing $\sigma$ and lying on the boundary of $J^+(\Sigma,\tilde{M})$. In particular, the future null line $\sigma'(t)=(0,1+t,t)$ is a null geodesic generator of $\partial P$. However, the pair $(P',\emptyset)\in \overline{M}$, $P'=I^-(\sigma')$ (which is associated with the point $(0,2,1)$) is not included in $\mathcal{J}^{+}$. Therefore $\mathcal{J}^{+}$ is not past-complete.
  \item There are no future-complete null C-rays, since any such null C-ray in $\tilde{M}$ must intersect $J^{+}(\Sigma,\tilde{M})$ (see Figure \ref{fig:3}).

    \end{itemize}
Thus, in order to violate the thesis of Theorem \ref{main}, we need to show that $C$ is not entirely contained inside the black hole.
Our construction, however still does not satisfy this property, since all causal curves emanating from $C$ intersect $J^{+}(\Sigma,\tilde{M})$, whose boundary is not contained in $\mathcal{J}^+$.

To complete the construction, consider then a conformally rescaled metric $\mathfrak{g}=\Omega\, g$ such that: (i) all the boundary points in $J^+(\Sigma,\tilde{M})\setminus\Sigma$ are in $\mathcal{J}^{+}$ with $\mathfrak{g}$ and (ii) all the null $C$-rays are still incomplete. In order to obtain an appropriate conformal factor $\Omega$, recall a classic result due to Clarke \cite{Clarkegeodesiccompletenesscausal1971} ensuring the existence of a conformal factor $\tilde{\Omega}$ such that all lightlike geodesics on $(M,\tilde{\Omega} g)$ are complete. Now define $\Omega$ thus: (i) $\Omega\equiv 1$ in $I^{-}_{1/4}(\Sigma)$ and (ii) $\Omega\equiv \tilde{\Omega}$ in $M\setminus I^{-}(\Sigma)$, where $I^{-}_{a}(\cdot)$ denotes the chronological past computed with the metric $-adt^2+dx^2+dy^{2}$ (see Figure \ref{fig:4}).
This conformal factor ensures that all the future-directed lightlike geodesics with endpoint in $J^+(\Sigma,\tilde{M})\setminus \Sigma$ (and so, contained in a region where $\Omega\equiv \tilde{\Omega}$) are complete. In particular, $\left(J^+(\Sigma,\tilde{M})\setminus \Sigma\right)\subset \mathcal{J}^+$.

The situation with $\Sigma$ is quite peculiar: there exist future-complete null rays with endpoints on $\Sigma$ (so condition (i) on Defn. \ref{scri} holds) but not all null geodesics with endpoints on $\Sigma$ are complete (consider for instance the null rays emanating from $C$.) Hence, $\Sigma$ does not intersect $\mathcal{J}^+$ \textit{just because condition (ii) on Defn. \ref{scri} fails}. But now, since $\Sigma$ does not intersect $\mathcal{J}^+$, $(M,\Omega\,g)$ in turn fails to be past-complete, just by the same previous arguments. However, all the points in $C$ are now visible, as we can connect them with boundary points of $\left(J^+(\Sigma,\tilde{M})\setminus \Sigma\right)$ by means of a future-directed timelike curves.

\begin{figure}
\centering
\ifpdf
  \setlength{\unitlength}{1bp}%
  \begin{picture}(216.58, 213.83)(0,0)
  \put(0,0){\includegraphics{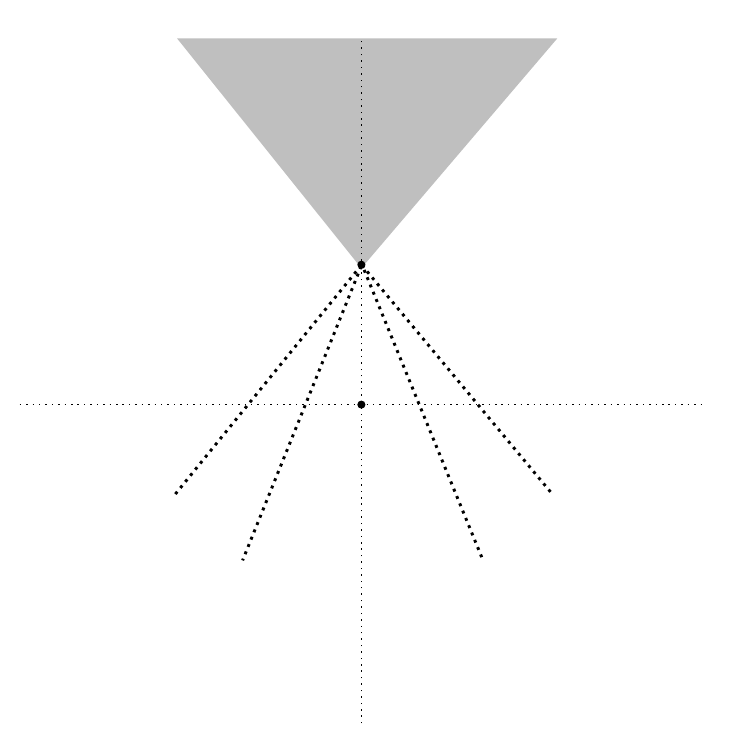}}
  \put(90.52,48.36){\fontsize{13.38}{16.86}\selectfont $\Omega\equiv 1$}
  \put(190.52,100.36){\fontsize{11.38}{14.86}\selectfont $\mathbf{x}$}
     \put(104.52,205.36){\fontsize{11.38}{14.86}\selectfont $\mathbf{t}$}
  \put(35.76,133.07){\fontsize{13.38}{16.86}\selectfont $\Omega=\tilde{\Omega}$}
  \end{picture}%
\else
  \setlength{\unitlength}{1bp}%
  \begin{picture}(216.58, 213.83)(0,0)
  \put(0,0){\includegraphics{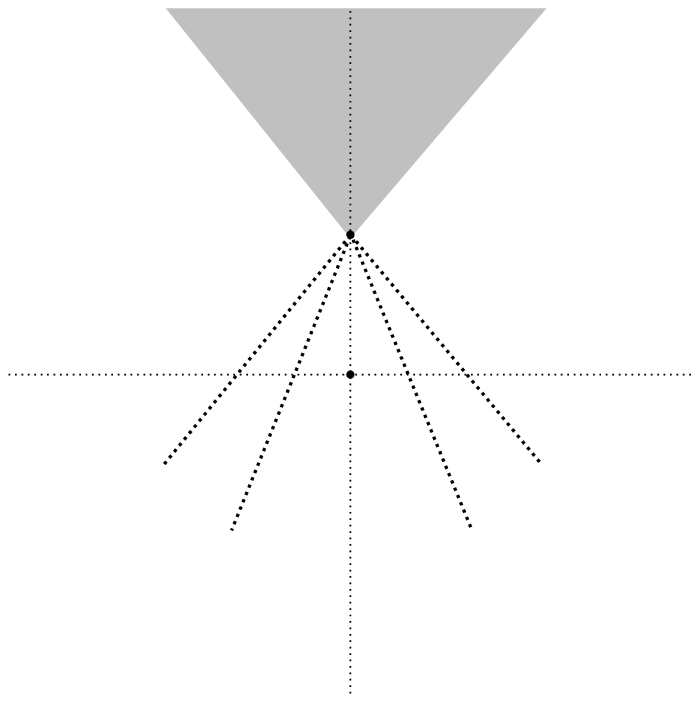}}
  \put(94.52,52.36){\fontsize{7.38}{8.86}\selectfont $\Omega\equiv 1$}
  \put(36.76,133.07){\fontsize{7.38}{8.86}\selectfont $\Omega=\tilde{\Omega}$}
  \end{picture}%
  \fi

  \caption{\label{fig:4} Illustration of the behaviour of the conformal factor $\Omega$ in the section $y=0$ of the spacetime.}
\end{figure}

\begin{remark}
 \emph{Observe that the previous example also shows that condition (ii) on Definition \ref{scri} too necessary to obtain Theorem \ref{main}. In fact, as we have mentioned, if we remove such a condition then the points of $\Sigma$ will also belong to $\mathcal{J}^{+}$, and the spacetime is actually past-complete. However, it is still true that there are no future-complete null $C$-rays. }
\end{remark}

\section*{Acknowledgments}

The authors are partially supported by the Spanish Grant MTM2016-78807-C2-2-P (MINECO and FEDER funds). They wish to acknowledge the IEMath-GR, and especially Miguel S\'anchez, for the kind hospitality while part of the work on this paper was being carried out.

\vspace*{\fill}


\vspace{.5cm}
\end{document}